\documentclass[11pt,letterpaper]{article}

\usepackage[margin=1in]{geometry}
\usepackage{latexsym,graphicx,amssymb}
\usepackage{amsmath}
\usepackage{amsthm, amsfonts}
\usepackage{mathtools, thmtools}
\usepackage{bbm}
\usepackage{float}
\usepackage{xspace}
\usepackage{paralist}
\usepackage{enumerate,multicol}
\usepackage{cases}
\usepackage{caption}
\usepackage{graphicx}
\usepackage{tikz,pgfmath}
\usepackage{pdfcomment}
\usetikzlibrary{matrix,positioning,quotes}
\usetikzlibrary{shapes.multipart}
\usetikzlibrary{calc}
\usetikzlibrary{arrows,decorations.markings}
\usetikzlibrary{matrix}
\usepackage{url}
\usepackage[ruled,linesnumbered,vlined]{algorithm2e}
\usepackage{hyperref}
\hypersetup{hidelinks}
\usepackage[capitalise]{cleveref}

\usepackage[noend]{algpseudocode}
\usepackage{xcolor}
\usepackage{multirow}
\usepackage{graphicx}
\usepackage{subcaption}
\usepackage{varwidth}
\usepackage{wrapfig}
\usepackage{enumitem}
\usepackage{float}
\usepackage{bbm,bm}
\usepackage{tcolorbox}

\usepackage{mleftright}
\usepackage[d]{esvect}
\usepackage{bm, bbm}

\newenvironment{proofof}[1]{\begin{proof}[Proof of #1]}{\end{proof}}
\numberwithin{figure}{section}
\newtheorem{theorem}{Theorem}[section]
\newtheorem{definition}[theorem]{Definition}

\newtheorem{lemma}[theorem]{Lemma}
\newtheorem{claim}[theorem]{Claim}
\newtheorem*{theorem*}{Theorem}
\newtheorem*{lemma*}{Lemma}

\crefname{Distribution}{Distribution}{Distributions}
\crefname{claim}{Claim}{Claims}
\Crefname{claim}{Claim}{Claims}

\input{head}

\newcommand{\bk}[1]{\left(#1\right)}

\newcommand{\bigbk}[1]{\bigl( {#1} \bigr)}
\newcommand{\Bigbk}[1]{\Bigl( {#1} \Bigr)}
\newcommand{\Bk}[1]{\left[{#1}\right]}
\newcommand{\bigBk}[1]{\bigl[ {#1} \bigr]}
\newcommand{\BigBk}[1]{\Bigl[ {#1} \Bigr]}
\newcommand{\BK}[1]{\left\{{#1}\right\}}
\newcommand{\midBK}[1]{\{ {#1} \}}
\newcommand{\bigBK}[1]{\bigl\{ {#1} \bigr\}}

\newcommand{\abs}[1]{\left|{#1}\right|}
\newcommand{\bigabs}[1]{\big|{#1}\big|}

\newcommand{\E}{\mathop{\mathbb{E}}}
\renewcommand{\Pr}{\mathop{\mathrm{Pr}}}
\newcommand{\poly}{\mathrm{poly}}

\renewcommand{\l}{\ell}
\newcommand{\eps}{\varepsilon}

\renewcommand{\vec}[1]{\bm{\mathrm{#1}}}
\renewcommand{\emptyset}{\varnothing}

\newcommand{\defeq}{\coloneqq}

\newcommand{\bef}{\textup{bef}}

\newcommand{\inseta}{A}
\newcommand{\insetb}{B} %
\newcommand{\delset}{D} %
\newcommand{\keyset}{K}
\newcommand{\keyend}{\keyset_\textup{end}}

\newcommand{\insas}{\vec{a}} %
\newcommand{\insai}[1][i]{a_{#1}} %
\newcommand{\insbs}{\vec{b}} %
\newcommand{\insbi}[1][i]{b_{#1}} %
\newcommand{\dels}{\vec{d}} %
\newcommand{\deli}[1][i]{d_{#1}} %

\newcommand{\insasi}[1][i]{\vec{a}^{[#1]}} %
\newcommand{\insbsi}[1][i]{\vec{b}^{[#1]}} %
\newcommand{\delsi}[1][i]{\vec{d}^{[#1]}} %

\newcommand{\insaset}{\mathcal{A}} %
\newcommand{\insaseti}[1][i]{A^{[#1]}} %
\newcommand{\insaparts}{\insaseti[1], \ldots, \insaseti[\br]} %

\newcommand{\insrest}{\insaset_\textup{rest}} %

\newcommand{\insorder}{\pi} %
\newcommand{\delorder}{\sigma} %
\newcommand{\insorders}{\vec{\insorder}} %
\newcommand{\insorderi}[1][i]{\insorder^{[#1]}} %
\newcommand{\insorderij}[2][i]{\insorder^{[#1]}_{#2}} %

\newcommand{\sta}{C} %
\newcommand{\staof}[1][i]{\sta_{#1}} %
\newcommand{\stabegin}{\staof[\textup{st}]} %
\newcommand{\staend}{\staof[\textup{end}]} %
\newcommand{\stabef}{\staof[\bef]} %
\newcommand{\staa}{\staof[\inseta]} %
\newcommand{\stab}{\staof[\insetb]} %
\newcommand{\staai}{\staof[\inseta]^{[i]}} %
\newcommand{\stamix}{\staof[\textup{mix}]} %

\newcommand{\lv}{\ell} %
\newcommand{\wid}{m} %
\newcommand{\widof}[1][\lv]{\wid_{#1}} %
\newcommand{\widl}{\widof} %
\newcommand{\widn}{\widof[\lv-1]} %
\newcommand{\br}[1][\lv]{\lambda_{#1}} %
\newcommand{\branch}{\lambda} %

\newcommand{\set}{S} %
\newcommand{\setof}[1][A]{\set_{#1}} %
\newcommand{\seta}{\setof[\inseta]} %
\newcommand{\setai}{\setof[\inseta]^{[i]}} %
\newcommand{\setb}{\setof[\insetb]} %
\newcommand{\setbad}{\setof[\textup{bad}]} %
\newcommand{\setab}{\seta\setminus\setb} %

\newcommand{\cost}{\textup{\textsf{cost}}} %
\newcommand{\costu}{\cost_u} %
\newcommand{\probe}{\textup{\textsf{probe}}} %
\newcommand{\probeu}{\probe_u} %

\newcommand{\guess}{\mathcal F} %
\newcommand{\guessall}{\mathcal F^{[i]}_{\textup{all}}} %
\newcommand{\guesspass}{\mathcal F^{[i]}_{\textup{qual}}} %
\newcommand{\guessnbr}{\mathcal F^{[i]}_{\textup{nbr}}} %

\newcommand{\redun}{r} %
\newcommand{\Redun}{R} %

\newcommand{\msgout}{M_i} %
\newcommand{\indbad}{X_i} %
\newcommand{\indnbr}{Y_i} %
\newcommand{\indprobe}{Z_i} %
\newcommand{\indin}{W} %

\newcommand{\msgfrom}{M_\textup{send}} %
\newcommand{\msgto}{M_\textup{learn}} %

\newcommand{\indicator}[1]{\mathbbm{1}[#1]}
\newcommand{\indicatorEx}[1]{\mathbbm{1}\Bk{#1}}
\newcommand{\mutual}[2]{I\bk{#1 \, ; \, #2}}

\newcommand*\oline[1]{%
  \vbox{%
    \hrule height 0.5pt%
    \kern0.25ex%
    \hbox{%
      \kern-0.1em%
      \ifmmode#1\else\ensuremath{#1}\fi%
      \kern-0.1em%
    }%
  }%
}

\newcommand{\setnota}{\hspace{0.1em}\oline{\seta}\hspace{0.1em}}
\newcommand{\setnotb}{\hspace{0.1em}\oline{\setb}\hspace{0.1em}}
\newcommand{\setnotbstar}{\hspace{0.1em}\oline{\setb^*}\hspace{0.1em}}
\newcommand{\setnotab}{\hspace{0.1em}\oline{\seta \cup \setb}\hspace{0.1em}}
\newcommand{\setnotacapb}{\hspace{0.1em}\oline{\seta \cap \setb}\hspace{0.1em}}

\author{
  Tianxiao Li \thanks{Institute for Interdisciplinary Information Sciences, Tsinghua University. \url{litx20@mails.tsinghua.edu.cn}.}
  \and
  Jingxun Liang \thanks{Institute for Interdisciplinary Information Sciences, Tsinghua University. \url{liangjx20@mails.tsinghua.edu.cn}.} 
  \and
  Huacheng Yu \thanks{Department of Computer Science, Princeton University. \url{yuhch123@gmail.com}.}
  \and
  Renfei Zhou \thanks{Institute for Interdisciplinary Information Sciences, Tsinghua University. \url{zhourf20@mails.tsinghua.edu.cn}.}
}
\title{Tight Cell-Probe Lower Bounds for Dynamic Succinct Dictionaries}
\date{}

\begin{document}

\maketitle

\thispagestyle{empty}
\setcounter{page}{0}

\begin{abstract}
	A dictionary data structure maintains a set of at most $n$ keys from the universe $[U]$ under key insertions and deletions, such that given a query $x\in[U]$, it returns if $x$ is in the set.
	Some variants also store values associated to the keys such that given a query $x$, the value associated to $x$ is returned when $x$ is in the set.

	This fundamental data structure problem has been studied for six decades since the introduction of hash tables in 1953. 
	A hash table occupies $O(n\log U)$ bits of space with constant time per operation in expectation.
	There has been a vast literature on improving its time and space usage.
	The state-of-the-art dictionary by Bender, Farach-Colton, Kuszmaul, Kuszmaul and Liu~\cite{liu22} has space consumption close to the \emph{information-theoretic optimum}, using a total of
	\[
		\log\binom{U}{n}+O(n\log^{(k)} n)
	\]
	bits, while supporting all operations in $O(k)$ time, for any parameter $k\leq \log^* n$.
	The term $O(\log^{(k)} n)=O(\underbrace{\log\cdots\log}_k n)$ is referred to as the \emph{wasted bits per key}.

	In this paper, we prove a matching \emph{cell-probe} lower bound: For $U=n^{1+\Theta(1)}$, any dictionary with $O(\log^{(k)} n)$ wasted bits per key must have expected operational time $\Omega(k)$, in the cell-probe model with word-size $w=\Theta(\log U)$.
	Furthermore, if a dictionary stores values of $\Theta(\log U)$ bits, we show that \emph{regardless of the query time}, it must have $\Omega(k)$ expected update time.
	It is worth noting that this is the first cell-probe lower bound on the trade-off between space and update time for general data structures.
\end{abstract}

\newpage

\section{Introduction}
\label{sec:intro}

A dictionary data structure dynamically maintains a set of at most $n$ keys from the universe $[U]$ under key insertions and deletions (the updates), such that given a query $x\in[U]$, it returns if $x$ is in the set.
Some variants of dictionaries also maintain a set of key-value pairs with distinct keys and with values from $[V]$, such that given a query $x$, it further returns the value associated to $x$ when key $x$ is in the set.

Hash tables are classic designs of dictionaries.
Using universal hash functions and chaining, hash tables have constant expected update and query times while occupying $O(n\log U)$ bits of space.
There has been a vast literature on dictionaries improving the time and space usage \cite{Knuth73, ANS09, ANS10, bender2021all, DdHPP06, PerfectHashing88, RealTime90, FPSS03, Storing84, knuth1963notes, LYY20, Pagh01, PAGH2004122, patrascu2008succincter, RR03, Yu20, liu22}.
An ideal dictionary would use nearly information-theoretically optimal space, $\approx\log\binom{U}{n}$ bits, and could process each operation in constant time.

The state-of-the-art dictionary by Bender, Farach-Colton, Kuszmaul, Kuszmaul and Liu~\cite{liu22} uses
\[
  \log\binom{U}{n}+O(n\log^{(k)} n)
\]
bits of space, while supporting each operation in $O(k)$ time for any parameter $k\in[\log^* n]$.
The term
\[
  O(\log^{(k)} n)=O(\underbrace{\log\log\cdots\log}_{k} n)
\]
is referred to as the \emph{wasted bits per key}, and $O(n\log^{(k)} n)$ is referred to as the \emph{redundancy}.
This interesting time-space trade-off may not look natural at the first glance, e.g., when $O(1)$ wasted bits per key are allowed, the data structure has operational time $O(\log^* n)$.
Surprisingly, they proved that for any hash table that makes use of ``augmented open addressing,'' which includes all known \emph{dynamic} succinct dictionaries \cite{RR03, ANS10, Bloom18, Bercea2020ADS, LYY20, bender2021all, liu22}, this trade-off is optimal!

This leads to the question of how general ``augmented open addressing'' is, and whether one can design a dynamic dictionary that does not fall in this category and has a better time-space trade-off.
For example, the best-known \emph{static} dictionary (i.e., a data structure that only needs to support queries)~\cite{Yu20}, which achieves $O(1)$ expected query time with $O(n^{\varepsilon})$ redundancy, does not use augmented open addressing.

In this paper, we show that this is indeed the best possible for \emph{dynamic} dictionaries, by proving a matching \emph{cell-probe} lower bound.\footnote{In the cell-probe model~\cite{Yao78}, it takes unit cost to read or write (probe) one memory cell of $w$ bits, and the computation is free. Thus, a cell-probe lower bound implies the same time lower bound for RAM.}
\begin{theorem}\label{thm_main}
  For $U=n^{1+\Theta(1)}$ and $k\leq \log^* n$, any dynamic dictionary storing at most $n$ keys from $[U]$ with $O(\log^{(k)} n)$ wasted bits per key must have expected insertion, deletion \emph{or} query time at least $\Omega(k)$, in the cell-probe model with word-size $w=\Theta(\log U)$.
\end{theorem}
It is worth noting that when $U/n$ is slightly sub-polynomial, e.g., $U=n^{1+1/\log^{(k)} n}$ for some constant $k$, Bender et al.~\cite{liu22} proposed another data structure with $o(1)$ wasted bits per key supporting constant-time insertions, deletions and queries.
Thus, the requirement that $U$ needs to be at least $n^{1+\Omega(1)}$ can barely be relaxed.

Our proof uses a framework similar to the \emph{information transfer tree} argument~\cite{PD04}, which is now a widely-used technique for proving cell-probe lower bounds~\cite{PD04b,PD06,CJ11,CJS15,CJS16,Yu16,WY16,AWY18,LN18,BHN19,JLN19,LMWY20}.
Roughly speaking, it builds a tree on top of a sequence of $n$ operations, and associates each cell-probe to an internal node of the tree.
The key step of the proof is to lower bound the number of cell-probes associated to each node.
In all prior work using this framework, this is done by reasoning about how the updates and queries in different subtrees must interact (i.e., to answer a query, the data structure must learn sufficient information about previous updates), thus proving a trade-off between update time and query time.
In our proof, we are able to reason, via a novel argument, about how \emph{space constraints} can force the operations to spend cell-probes, thus proving a time-space lower bound.

In fact, our new technique already gives an arguably simpler and more intuitive proof of the lower bound against data structures using augmented open addressing that Bender et al. proved.
We will present an overview of this simplified proof in the next section as a warm-up.

The technique also allows us to extend the lower bound to $o(n)$ redundancy.
\begin{theorem}\label{thm_main_2}
  For $U=n^{1+\Theta(1)}$, any dynamic dictionary storing at most $n$ keys from $[U]$ with $R<n$ bits of redundancy must have expected insertion, deletion \emph{or} query time at least $\Omega(\log (n/R))$, in the cell-probe model with word-size $w=\Theta(\log U)$.
\end{theorem}

Furthermore, if the keys are associated with values of $\Theta(\log U)$ bits, the same technique proves that even if the queries are allowed to take arbitrarily long time, the updates must still follow the same lower bound.
\begin{restatable}{theorem}{UpdateOnly}\label{thm:update_only}
  Consider a dynamic dictionary storing at most $n$ keys from $[U]$, each associated with a value in $[V]$, with $\Redun$ bits of redundancy, where $U \ge 3n$ and $V = U^{2 + \Theta(1)}/n^2$. Then, in the cell-probe model with word-size $w=\Theta(\log U)$,
  \begin{enumerate}
  \item if $\Redun \geq n$ can be written as $\Redun = O(n \log^{(k)} n)$ for $k\leq\log^*n$, then the expected update time is $\Omega(k)$;
  \item if $\Redun < n$, then the expected update time is $\Omega(\log (n / \Redun))$.
  \end{enumerate}
\end{restatable}
This lower bound shows the dictionary by Bender et al.~(which can support values) is optimal in a very strong sense: their data structure achieves $O(k)$ \emph{update time} in worst-case except with inverse polynomial probability,\footnote{By simply rebuilding the whole data structure when an update takes more than $\Omega(k)$ time, the worst-case-with-high-probability bound implies $O(k)$ expected time.} $O(1)$ \emph{query time} in worst-case, and it is dynamically resizable;\footnote{A dictionary is dynamically resizable if its space usage is in terms of the \emph{current} data size.}
we show that even if we relax all other properties, the $O(k)$ update time is still not improvable just under the space constraint.

We emphasize that this is the first space-update trade-off lower bound in the cell-probe model for any data structure problem.
Such lower bounds were proved in~\cite{LNN15} for streaming problems against \emph{non-adaptive} update algorithms, which was later shown not to hold for general adaptive update algorithms~\cite{AY20}.
In fact, by applying \emph{global rebuilding}~\cite{Overmars83}, one cannot hope to prove such a trade-off when the redundancy is linear in the total space without proving a \emph{super-linear RAM lower bound}, a notoriously hard question.
Naively, a data structure can always use half of its memory as a buffer to store the unprocessed updates, and only batch-process it when it is full; if there is a batch-processing algorithm with linear time in \emph{RAM}, then this data structure directly has amortized update time $O(1)$ -- proving an $\omega(1)$ update time lower bound implies a RAM lower bound for batch processing. 
Furthermore, to achieve non-amortized update time $O(1)$, global rebuilding suggests building \emph{two} buffers in half of the memory, and using them alternatingly: when one buffer is full, one could gradually execute the batch-processing algorithm during the next $O(n)$ updates, and use the other buffer to store them.
Note that such a strategy is not applicable in the \emph{succinct} regime, i.e., where the redundancy is asymptotically $o(1)$-fraction of the total space -- the size of the buffers cannot exceed the redundancy, thus we would have to flush the buffers too frequently.

\bigskip

With an easy adaption, our technique can also prove tight lower bounds for strongly history-independent dictionaries. A dictionary is said to be \emph{strongly history-independent} (a.k.a.~\emph{uniquely representable}) if its memory state is fully determined by the set of elements in it together with some random bits the algorithm uses. We will prove the following result.

\begin{restatable}{theorem}{HistIndLB}\label{thm:lb_hist_ind}
  For $U = n^{1 + \Theta(1)}$ and $R \ge 1$, any strongly history-independent dynamic dictionary storing at most $n$ keys from $[U]$ with $R$ bits of redundancy must have expected insertion, deletion \emph{or} query time at least $\Omega\bigbk{\log \frac{n \log U}{R}}$, in the cell-probe model with word-size $w = \Theta(\log U)$.
\end{restatable}

This lower bound is tight for $R \ge n / \poly \log n$, due to the known data structures \cite{Kus23, LiL23}. Note that this time-space trade-off is worse than the optimal trade-off for succinct dynamic dictionaries without the history-independent constraint. We remark that this is the first separation between a data structure problem and its strongly history-independent version under the RAM (or cell-probe) model.

\smallskip

Another related problem is \emph{Stateless Allocation} \cite{Gol08, NaoT01, GooK17, BerK20, Kus23}. It requires the algorithm to put at most $(1-\eps) n$ elements from universe $[U]$ into $n$ slots, where each slot can only accommodate one element, and $\eps n$ slots are left empty. We require the mapping between elements and slots to only depend on the current element set and some random bits, i.e., the assignment is strongly history-independent. The performance of the algorithm is measured by its \emph{switching cost}, the number of elements that changes its assigned slot during an insertion/deletion. Kuszmaul~\cite{Kus23} showed an algorithm with expected switching cost $O(\log \eps^{-1})$ for $1/n \le \eps \le 1$. We will prove that this upper bound is actually tight.

\begin{restatable}{theorem}{AllocLB}\label{thm:lb_alloc}
  For $1/n \le \eps \le 1$ and $U \ge 3n$, any stateless allocation algorithm that assigns at most $(1-\eps) n$ elements from universe $[U]$ to $n$ slots must have expected switching cost at least $\Omega(\log \eps^{-1})$.
\end{restatable}

\section{Technical Overview}
\label{sec:overview}

Now we present an overview of our proof technique.
For simplicity, let us focus on the case with $O(n)$ bits of redundancy, and show a time lower bound of $\Omega(\log^* n)$.
Let us also assume $U=V=n^3$.

\subsection{Slot Model}
As a ``proof-of-concept,'' we first showcase a lower bound in the \emph{slot model}, a generalization of augmented open addressing.
Then we move on to discussing how to adapt the proof to the cell-probe model in \cref{sec:overview_cell_probe}, before presenting the formal proof in the later sections.

In the slot model, a data structure maintains a set of at most $n$ keys (balls) from $[U]$ under key insertions and deletions, and maps the keys to $n$ slots (bins) indexed by $[n]$, with at most one key in each slot at any point.
The data structure must maintain at most $\log\binom{U}{n}+O(n)$ bits of memory, which encodes the set itself and $O(n)$ bits of (arbitrary) auxiliary information, and can also determine the location (i.e., corresponding slot) of each key.
We say a key $k$ ``is in slot $i$'' if it is mapped to slot $i$ (it is irrelevant to how $k$ is stored in the memory state).
Each time we insert or delete a key, the data structure updates its memory and moves (or swaps) the keys between slots.
In the slot model, moving a key from one slot to another takes $O(1)$ cost, while accessing the memory is free.
Thus, the goal is to minimize the key-moves during the insertions and deletions.
We will be proving that each operation must move at least $\Omega(\log^* n)$ keys on average.

We note that the slot model is similar to, yet more general than, the augmented open addressing defined in~\cite{liu22}.
Augmented open addressing has a list of ``hash functions'' $h_1,h_2,\ldots,h_n$ mapping $[U]$ to $[n]$.
Key $x$ must be put in slot $h_i(x)$ for some $i$, and the auxiliary information stores the index $i$ using $O(\log i)$ bits.

One can either view a data structure in the slot model as a dictionary with values, where the balls represent the keys' associated values (which are physically stored in their corresponding slots), or view the ``quotient''~\cite{Knuth73,Pagh99} as the balls in the key-only setting.
It is more limited than general data structures as the values or quotients must be stored in $n$ memory slots atomically, while it is also stronger in the sense that the set of keys is known for free and we only count the number of moves.

\paragraph{Hard distribution.}
Let us consider the following sequence of operations (similar to the hard instance analyzed in~\cite{liu22}): We initialize by inserting a set of $n$ random keys, then repeatedly delete a random key from the initial set and insert a new random key from $[U]$.
We call each pair of deletion and insertion after the initialization a \emph{meta-operation}.
The analysis will focus on the cost of the $n$ meta-operations.

\paragraph{A simple upper bound.}
To motivate the quantity $\log^* n$, let us first consider the following algorithm that maintains the keys approximately in sorted order using lazy updates.
After the initialization, the data structure spends $O(n)$ swaps to place the keys in sorted order in the $n$ slots.
Thus, the key set by itself determines the location of every key.

Next, for each meta-operation, after deleting a key and thus emptying a slot, we can first put the new key in the empty slot, and store its identity and location as auxiliary information.
As we delete and insert more keys, the slots gradually become ``less sorted.''
Since we only allow $O(n)$ bits of auxiliary information, this can be done as long as the new keys and their locations take at most $O(n)$ bits to store, i.e., we can store $O(n/\log n)$ new keys in this way.

After $O(n/\log n)$ meta-operations, we rearrange all \emph{new keys} so that they become in sorted order \emph{by themselves}, while the set of new keys still resides in the set of slots that contained deleted keys (see the first three steps in \Cref{fig:key_arrangement} for an example).
This takes $O(1)$ \emph{amortized} swaps per insertion.
It then suffices to only indicate as auxiliary information which keys are the new keys and the \emph{set} of their locations.
This takes $\log\binom{n}{O(n/\log n)}=O(n\log\log n/\log n)$ bits.
After the rearrangement, every new key effectively only costs $O(\log\log n)$ auxiliary bits. 

\begin{figure}[t]
\begin{tikzpicture}
	\node[] at (0,0) {$1$};
	\node[] at (.5,0) {$2$};
	\node[] at (1,0) {$3$};
	\node[] at (1.5,0) {$4$};
	\node[] at (2,0) {$5$};
	\draw [-] (-.4,-.3) rectangle (2.4,.3);
	\draw [-latex] (1.5,.6) -- (1.5,.3);
	\node[] at (1.5,.9) {$6$};
	
	\draw [->>] (2.5, 0) -- (4.5, 0);
	\node[] at (3.5,.3) {replace 4};
	\node[] at (3.5,-.3) {with 6};

	\node[] at (5,0) {$1$};
	\node[] at (5.5,0) {$2$};
	\node[] at (6,0) {$3$};
	\node[draw] at (6.5,0) {$6$};
	\node[] at (7,0) {$5$};
	\draw [-] (4.6,-.3) rectangle (7.4,.3);
	\draw [-latex] (5.5,.6) -- (5.5,.3);
	\node[] at (5.5,.9) {$7$};
	
	\draw [->>] (7.5, 0) -- (10, 0);
	\node[] at (8.75,.3) {after $n/\log n$};
	\node[] at (8.75,-.3) {operations};
	
	\node[] at (6,-.6) {aux info: $O(\log n)$};

	\node[] at (10.5,0) {$1$};
	\node[draw] at (11,0) {$7$};
	\node[] at (11.5,0) {$3$};
	\node[draw] at (12,0) {$6$};
	\node[] at (12.5,0) {$5$};
	\draw [-] (10.1,-.3) rectangle (12.9,.3);
	
	\draw [->>] (13, 0) -- (13.6, 0) -- (13.6, -2) -- (13, -2);
	\filldraw[fill=white, draw=white] (13.5, -.5) rectangle (13.7, -1.4);
	\node[] at (14.3,-.7) {rearrange last};
	\node[] at (14.3,-1.2) {$n/\log n$ keys};
	
	\node[] at (11.5,-.6) {aux info:};
	\node[] at (11.5,-1.1) {$O(\frac{n}{\log n}\cdot \log n)$};

	\node[] at (10.5,-2) {$1$};
	\node[draw] at (11,-2) {$6$};
	\node[] at (11.5,-2) {$3$};
	\node[draw] at (12,-2) {$7$};
	\node[] at (12.5,-2) {$5$};
	\draw [-] (10.1,-2.3) rectangle (12.9,-1.7);
	\draw [latex-latex](11,-1.6) parabola bend (11.5,-1.5) (12,-1.6);
	
	\node[] at (11.5,-2.6) {aux info:};
	\node[] at (11.5,-3.1) {$O(\log\binom{n}{n/\log n})$};

	\node[draw,circle,inner sep=1.5] at (5,-2) {$9$};
	\node[draw] at (5.5,-2) {$6$};
	\node[] at (6,-2) {$3$};
	\node[draw] at (6.5,-2) {$7$};
	\node[draw,circle,inner sep=1.5] at (7,-2) {$8$};
	\draw [-] (4.6,-2.3) rectangle (7.4,-1.7);
	\draw [-latex] (5,-1.4) -- (5,-1.7);
	\draw [-latex] (7,-1.4) -- (7,-1.7);
	
	\draw [<<-] (7.5, -2) -- (10, -2);
	\node[] at (8.75,-1.7) {$n/\log \log n$};
	\node[] at (8.75,-2.3) {operations};
	
	\draw [->>] (4.5, -2) -- (1,-2) --(1,-2.7);
	\node[] at(2.75,-1.7) {rearrange last};
	\node[] at(2.75,-2.3) {$n/\log n$ keys};

	\node[draw,circle,inner sep=1.5] at (0,-3.5) {$8$};
	\node[draw] at (.5,-3.5) {$6$};
	\node[] at (1,-3.5) {$3$};
	\node[draw] at (1.5,-3.5) {$7$};
	\node[draw,circle,inner sep=1.5] at (2,-3.5) {$9$};
	\draw [-] (-.4,-3.8) rectangle (2.4,-3.2);
	\draw [latex-latex](0,-3.1) parabola bend (1,-2.9) (2,-3.1);
	\node[] at (1,-4.1) {aux info:};
	\node[] at (1,-4.6) {$O(\frac{\log n}{\log\log n}\cdot\log\binom{n}{n/\log n})$};
	
	\draw [->>] (2.5, -3.5) -- (5.5, -3.5);
	\node[] at (4,-3.2) {rearrange last};
	\node[] at (4,-3.8) {$n/\log\log n$ keys};

	\node[draw] at (6,-3.5) {$6$};
	\node[draw] at (6.5,-3.5) {$7$};
	\node[] at (7,-3.5) {$3$};
	\node[draw] at (7.5,-3.5) {$8$};
	\node[draw] at (8,-3.5) {$9$};
	\draw [-] (5.6,-3.8) rectangle (8.4,-3.2);
	\node[] at (7,-4.1) {aux info:};
	\node[] at (7,-4.6) {$O(\log\binom{n}{n/\log\log n})$};
\end{tikzpicture}
\caption*{\footnotesize \raggedright In the first $n/\log n$ operations, the new keys are directly placed in the empty slots which are marked by squares (the first two steps), and will be rearranged later (the third step). This process continues for each subsequent group of $n/\log n$ operations, until the $(\log n/ \log \log n)$-th group, which is marked by circles (the fourth step). After rearranging this group (the fifth step), we further rearrange all the $n/\log \log n$ new keys (the sixth step). }
\caption{Key Arrangement}
\label{fig:key_arrangement}
\end{figure}

We do this to every segment of $O(n/\log n)$ consecutive meta-operations: store their identities and locations until it takes $O(n)$ bits; then rearrange so that they become in sorted order by themselves; and further store for every new key and every emptied slot, which segment we inserted the key or emptied the slot.
Since every new key and slot needs $O(\log\log n)$ bits to encode, we can repeat this for $O(\log n/\log\log n)$ segments, i.e., $O(n/\log\log n)$ meta-operations in total.
Thereafter, we further sort these $O(n/\log\log n)$ keys, as shown in the last step in \Cref{fig:key_arrangement}. 
Now we only need to pay $O\bigbk{\log\binom{n}{n/\log\log n}}=O(n\log\log\log n/\log\log n)$ total auxiliary bits, or $O(\log\log\log n)$ auxiliary bits per key.
Sorting again takes amortized $O(1)$ additional swaps per meta-operation.

This suggests the following strategy based on lazy updates:
For each meta-operation, we can place the new key in the slot that just became empty and store its identity and location;
after every segment of $O(n/\log n)$ consecutive meta-operations, we sort the new keys inserted in this segment and store the set of locations as a batch\footnote{Strictly, we also store which keys belong to the new batch, whose space usage is no more than storing the set of locations.};
after every $O(\log n/\log\log n)$ consecutive segments, we sort all new keys inserted in them and store the set of locations, and so on.
This process continues until we have inserted $O(n)$ keys, in which case, we can afford to re-sort all keys with $O(1)$ additional swaps per key.
There will be a total of $O(\log^*n)$ levels of sorting.
Since sorting in each level takes $O(1)$ swaps per key, the total number of swaps for $O(n)$ meta-operations is $O(n\log^*n)$.

\subsection{Slot Model Lower Bound}
The slot model lower bound is inspired by the above algorithm.
Note that in the algorithm, for every $O(n/\log n)$ consecutive meta-operations, we must rearrange the keys, otherwise the auxiliary information will take more than $O(n)$ bits.
The rearrangement takes $O(1)$ cost per key.
We will show that this $O(1)$-cost per key cannot be avoided.

We say a slot is \emph{accessed} if we move a key from or to this slot.
Formally, we will show that for every consecutive $m=\Omega(n/\log n)$ meta-operations, there must be $\Omega(m)$ times in expectation that a slot is accessed during multiple meta-operations.
In other words, we go over all $m$ meta-operations, and examine during each meta-operation, which slots are accessed \emph{both in this and in a previous meta-operation}, and add the number of such slots for every meta-operation together.
We assert that there must be $\Omega(m)$ times in expectation that a slot is accessed in a meta-operation and also in a previous one (thus, if the same slot is accessed during $t$ different meta-operations, it is counted $t-1$ times).

Similarly, for every consecutive $m=\Theta(n/\log\log n)$ meta-operations, consisting of $\Theta(\log n/\log\log n)$ segments of size $\Theta(n/\log n)$, we assert that there must be $\Omega(m)$ times in expectation that a slot is accessed during multiple \emph{segments}, and so on.
We formally state it in the following lemma, which is the slot-model variant of \cref{lm:outer}.

\begin{lemma*}
	Let $1\leq \ell\leq \log^* n$.
	For any consecutive $m=c_1\cdot n/\log^{(\ell)} n$ meta-operations for a sufficiently large constant $c_1$, partitioned into $\lambda=\Theta(\log^{(\ell-1)} n/\log^{(\ell)} n)$ \emph{segments} of size $\Theta(n/\log^{(\ell-1)}n)$, there must be $\Omega(m)$ times that a slot is accessed during multiple \emph{segments}.
\end{lemma*}

We used the convention that $\log^{(0)} n:=n$.
To see why it implies a lower bound of $\Omega(\log^* n)$, we build a tree on top of the $n$ meta-operations with depth $\log^* n$.
Each leaf of the tree (level-$0$ node) corresponds to one meta-operation, and we place the leaves in the same order as the meta-operations.
Then we divide the leaves into groups of $\Theta(n/\log n)$ consecutive meta-operations, and assign a common parent (a new node) at level $1$ to each group.
Next, $\Theta(\log n/\log\log n)$ consecutive level-$1$ nodes are grouped together, and assigned a common parent at level $2$, and so on.
In general, we group $\Theta(\log^{(\ell-1)} n/\log^{(\ell)} n)$ nodes at level $\ell-1$, and assign them a common parent at level $\ell$. See \Cref{fig:tree_structure} for details.

\begin{figure}[t]
\begin{tikzpicture}

\node[] at (0,0) {Operations:};
\node[] (p01) at (2,0) {$\textsf{op}_1$};
\node[] (p02) at (3,0) {$\textsf{op}_2$};
\node[] at (4,0) {$\cdots$};
\node[] (p03) at (5,0) {$\textsf{op}_{m_1}$};
\node[] (p11) at (6.5,0) {$\textsf{op}_{m_1+1}$};
\node[] (tmp) at (7.5,0) {$\cdots$};
\node[] (p12) at (8.5,0) {$\textsf{op}_{2m_1}$};
\node[] at (11,0) {$\cdots\qquad\cdots$};
\node[] (p00) at (13,0) {$\textsf{op}_{n-1}$};
\node[] (asd) at (14,0) {$\textsf{op}_{n}$};
\node[color=red] (t01) at (3,-.9) {probe cell $i$};
\draw[-latex,color=red] (t01) -- (p02);
\draw[-,color=red] (2.7,-.2) -- (3.3,-.2);
\node[color=red] (t02) at (8.5,-.9) {probe cell $i$};
\draw[-latex,color=red] (t02) -- (p12);
\draw[-,color=red] (8,-.2) -- (8.9,-.2);

\node[] at (0,1.5) {level 1:};
\node[draw] (p21) at (3.5,1.5) {$m_1=c_1n/\log n$ ops};
\draw[-] (p21) -- (p01);
\draw[-,color=red,line width=1.5] (p21) -- (p02);
\draw[-] (p21) -- (p03);
\node[draw] (p22) at (7.5,1.5) {$m_1$ ops};
\draw[-] (p22) -- (p11);
\draw[-,color=red,line width=1.5] (p22) -- (p12);
\draw[-] (p22) -- (tmp);
\node[draw] (p31) at (10.5,1.5) {$m_1$ ops};
\node [] at (9,1.5) {$\cdots$};
\node[draw] (p32) at (13,1.5) {$m_1$ ops};
\draw[-] (p32) -- (asd);
\draw[-] (p32) -- (p00);
\draw[-] (p32) -- (12,.3);
\draw [dashed](2.3,0.8) parabola bend (3.5,0.6) (4.7,0.8);
\node[] at (1.5,0.7) {$\lambda_1=\frac{c_1n}{\log n}$};
\draw [dashed](6.7,0.8) parabola bend (7.5,0.65) (8.3,0.8);
\node[] at (6.3,0.7) {$\lambda_1$};
\draw [dashed](12,0.8) parabola bend (13,0.65) (14,0.8);
\node[] at (11.6,0.7) {$\lambda_1$};

\node[] at (0,3) {level 2:};
\node[draw] (p41) at (5.75,3) {$m_2=c_1 n/\log \log n$ ops};
\node[] at (9,3) {$\cdots$};
\node[draw] (p42) at (11.5,3) {$m_2$ ops};
\draw[-,color=red,line width=1.5] (p41) -- (p21);
\draw[-,color=red,line width=1.5] (p41) -- (p22);
\draw[-] (p42) -- (p31);
\draw[-] (p42) -- (p32);
\draw [dashed](3.9,2.3) parabola bend (5.5,2.1) (7.2,2.3);
\node[] at(3,2.2) {$\lambda_2=\frac{\log n}{\log \log n}$};
\draw [dashed](10.7,2.2) parabola bend (11.7,2.1) (12.6,2.2);
\node[] at(10.3,2.2) {$\lambda_2$};

\node[color=red] (qwq) at (2.5,3.5) {LCA};
\draw[rounded corners,-latex,color=red] (qwq) -- (2.5,3) -- (3.5,3);

\node[] at (0,4.4) {$\vdots$};
\node[] at (0,5.5) {level $\log^* n$:};
\node[draw] (p51) at (8,5.5) {$n$ ops};
\draw[-] (p51)--(6,4.5);
\node[] at (6,4) {$\vdots$};
\draw[-] (p51)--(7,4.5);
\draw[-] (p51)--(10,4.5);
\node[] at (10,4) {$\vdots$};

\end{tikzpicture}
\caption{Tree Structure}
\label{fig:tree_structure}
\end{figure}
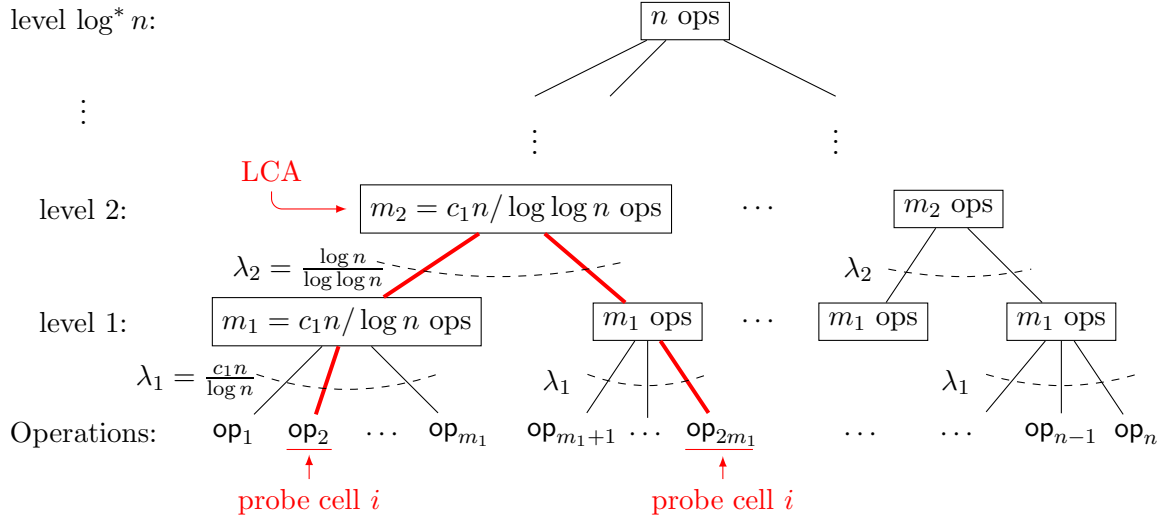

Now consider an access to a slot during a meta-operation (a leaf), and consider when it was accessed the \emph{previous time} (another leaf), we associate this access to the \emph{lowest common ancestor} (LCA) of these two leaves.
To lower bound the total number of accesses, it suffices to sum over all (internal) nodes the number of accesses associated to it.
On the other hand, whenever an access is associated to a node $u$, then with respect to $u$, this is a slot that is accessed in two different segments (two different children of $u$).
Hence, the above lemma lower bounds the number of accesses associated to each $u$ by the size of the subtree rooted at $u$.
Thus, the sum of every level is $\Omega(n)$, and the total sum over all internal nodes is $\Omega(n\log^* n)$.
This way of counting the number of accesses has been used in many cell-probe lower bounds~\cite{PD04b,PD06,CJ11,CJS15,CJS16,Yu16,WY16,AWY18,LN18,BHN19,JLN19,LMWY20}, although the nodes in the tree are usually set to the same degree, in which case, one obtains a logarithmic lower bound.

\subsection{Number of Swaps per Node}
Next, let us prove the above lemma.
For simplicity, we first focus on $\ell=1$, i.e., $m=c_1\cdot n/\log n$ and ``segments'' consisting of one meta-operation ($\lambda=m$).
To begin with, let us see what would go wrong if there was \emph{zero} slot accessed during multiple meta-operations, then we will extend this proof-by-contradiction to $\Omega(m)$ slots.

Intuitively, we are going to show that the data structure stores too much information about the \emph{order} of insertions (given the deletions) as auxiliary information, which corresponds to storing the new keys and locations in the upper bound.
Fix a sequence of operations sampled from the distribution, and fix an index $k$. We will analyze the $(k+1)$-th to the $(k+m)$-th meta-operation.

\paragraph{The outer game.}
To formalize this intuition, let us consider the following communication game between two players Alice and Bob, which we call the \emph{outer game}.
\begin{itemize}
	\item Both players get the memory state $\stabegin$ after the $k$-th meta-operation, of $\log\binom{U}{n}+O(n)$ bits.
	It determines the set of $n$ initial keys $K$ before the $m$ meta-operations, as well as the slot allocation.
	The players also both know the keys that are deleted and the order of the deletions $\dels$, and the (unordered) \emph{set} of all inserted keys $\insaset$.
	\item Alice further knows the order of insertions $\insas$, a permutation of $\insaset$.
	\item Alice sends one message to Bob, and their goal is for Bob to recover the permutation $\insas$.
\end{itemize}

\paragraph{A communication protocol.}
Next, let us consider a protocol for the outer game.
Alice's (only) message of this protocol will be an encoding of the memory state after the $m$ meta-operations $\staend$.
Since Bob knows the initial set $\keyset$ and the set of deleted keys and inserted keys, he also knows the set of keys $\keyend$ after the $m$ meta-operations.
On the other hand, $\staend$ also determines $\keyend$ -- the entropy of $\staend$ \emph{conditioned on} $\keyend$ is at most $O(n)$, since
\begin{align*}
	H(\staend\mid\keyend)&=H(\staend,\keyend)-H(\keyend)=H(\staend)-H(\keyend) \\&\leq \left(\log\binom{U}{n}+O(n)\right)-\log\binom{U}{n}=O(n).
\end{align*}
Thus, by encoding $\staend$ optimally \emph{conditioned on} $\keyend$, Alice can send $O(n)$ bits in expectation so that Bob is able to recover $\staend$ from the message.

Now we claim that Bob can fully recover the permutation $\insas$ if every slot is accessed in at most one meta-operation.
Let us consider the first insertion.
Bob knows the first deleted key $\deli[1]$, but only knows that the first inserted key $\insai[1]$ is among $\insaset$.
Without loss of generality, we may assume that the data structure always first puts $\insai[1]$ in the slot where $\deli[1]$ was, then immediately applies an arbitrary permutation over the slots.
Recall that any permutation can be decomposed into a collection of disjoint cycles: In particular, suppose $\deli[1]$ (and $\insai[1]$ temporarily) was in slot $s_1$, then there exist slots $s_2,\ldots,s_c$ for some $c\geq 1$ such that the key in $s_1$ is moved to $s_2$, the key in $s_2$ is moved to $s_3$, and finally, the key in $s_c$ is moved to $s_1$.
Now observe that all slots $s_1,\ldots,s_c$ are accessed during this meta-operation. By our assumption that every slot can be accessed in at most one meta-operation, these keys must remain in the same slots in the final state $\staend$.
Next, since Bob knows $\deli[1]$, he knows $s_1$, then he can examine which key $x$ is in slot $s_1$ in $\staend$.
If $x$ is \emph{not} in $\keyset$, then he knows immediately that it must be $\insai[1]$, as this slot is not accessed afterwards.
Otherwise, $x$ must be the key that was initially in $s_c$ and was moved to $s_1$.
Thus, Bob learns $s_c$ by examining the location of $x$ in initial memory state $\stabegin$.
Then Bob can further examine which key is in $s_c$ in $\staend$, and it must be the key that was initially in $s_{c-1}$.
Bob can trace this cycle by alternatingly examining the keys in the slots in $\stabegin$ and $\staend$.
Eventually, he is able to identify a key that is not in $\keyset$, which must be $\insai[1]$.

After knowing $\insai[1]$, Bob simply simulates the data structure for the first meta-operation by deleting $\deli[1]$ and inserting $\insai[1]$, then proceeds to the second meta-operation.
Note that this process does not require any communication.
Finally, Bob is able to recover the order of all insertions after Alice sends the single message consisting of an encoding of $\staend$.
The message has expected length $O(n)$, but the order of insertions has entropy $\Theta(m\log m)$, which yields a contradiction when $m=c_1n/\log n$ for sufficiently large constant $c_1$.

\paragraph{Extending to $\Omega(m)$ cost and level $\ell>1$.}
In general, if there are at most $0.01m$ times that a slot is accessed more than once, then except for at most $0.01m$ meta-operations, all slots accessed in other meta-operations are not accessed later.
We then simply apply the above argument to iteratively figure out the remaining $0.99m$ insertions. For the other $0.01m$, Alice can simply send Bob the inserted keys within $\insaset$ using $\log m$ bits each (and $O(m)$ bits indicating which rounds Bob needs to run the above algorithm on his own).
Alice's message has total length $0.01m\log m+O(n)$, which again yields a contradiction.

For levels $\ell>1$, the players' goal becomes for Bob to figure out that, for each insertion in $\insaset$, which \emph{segment} it belongs to (we give the order \emph{within each segment} to both players).
In this case, Bob needs to learn $m\log \lambda\approx c_1 n$ bits of information.
It turns out that the same protocol works, and in each segment, by applying the same strategy more carefully, Bob can again find the inserted keys by following the cycles of the permutations.
Thus, the lemma holds, and the lower bound in the slot model follows.

\subsection{Cell-Probe Lower Bound}\label{sec:overview_cell_probe}
The lower bound in the cell-probe model is proved using the same framework.
Instead of the slots, we will work with \emph{memory cells}.
In general, a data structure does not necessarily have one part to store the key set and auxiliary information, and another part for the slots.
Moreover, (the ``quotient'' of) a key do not necessarily occupy a complete memory cell, and may be encoded arbitrarily.
We need a more careful argument in this case.

Fix a data structure in the cell-probe model with word-size $w=\Theta(\log U)$ that uses $\log \binom{U}{n}+O(n)$ bits of space.
We prove the following lemma.
\begin{lemma*}
	Let $2\leq \ell\leq \log^* n$.
	For any consecutive $m=c_1\cdot n/\log^{(\ell)} n$ meta-operations for a sufficiently large constant $c_1$, partitioned into $\lambda=\Theta(\log^{(\ell-1)} n/\log^{(\ell)} n)$ \emph{segments} of size $\Theta(n/\log^{(\ell-1)}n)$, there must be $\Omega(m)$ times that a memory cell is \emph{probed} during multiple \emph{segments}.
\end{lemma*}

By the same argument as in the slot model using trees, this lemma implies an $\Omega(\log^* n)$ cell-probe lower bound.
It turns out that in order to prove it, the only part that we need to ``upgrade'' from the slot model is the step in the protocol of the outer game, where we showed that Bob can figure out which segment contains each insertion based (essentially) only on the starting and ending memory states.
For the slot model, this step uses the fact that (the quotient of) the keys are atomic and can only be stored in slots and moved between the slots.
For general data structures, we will have to do it differently.

\paragraph{Recovering the segments.}
Now let us fix $\ell\geq 2$ (for technical reasons, the proof does not work for $\ell=1$).
We again consider the outer communication game, where
\begin{itemize}
 	\item both players know the initial memory state $\stabegin$, the (unordered) {set} of $m$ insertions $\insaset$, the (ordered) sequence of $m$ deletions $\dels$;
 	\item Alice further knows the order of insertions $\insas$; 
 	\item Bob only knows the order of insertions \emph{within each segment} (e.g., Bob knows for each segment, the first insertion is the $\pi(1)$-th lexicographically smallest inserted key in this segment, the second is the $\pi(2)$-th smallest inserted key, etc, but without knowing the actual set of inserted keys);
 	\item their goal is to let Bob learn the set of inserted keys $\insaseti$ in each segment $i\in[\lambda]$.
\end{itemize}

For simplicity, we again will first derive a contradiction when \emph{no} cell was probed in more than one segment.
Similar to the slot model, Alice can send Bob the final memory state $\staend$ using $O(n)$ bits.
Bob will then first try to identify $\insaseti[1]$, the insertions in the first segment.
To this end, Bob decodes $\staend$ from the message, then for all possible $(m/\lambda)$-sets of insertions $B\subset \insaset$, Bob pretends that $\insaseti[1]=B$, and simulates the data structure on $B$ together with the deletions in the first segment, which are known to Bob.

Now observe that a necessary condition for $\insaseti[1]=B$ is that the set of cells probed in this simulation have the same new cell-contents as in $\staend$. This is because we assumed that no cell is probed in more than one segment, thus their contents must not be updated later.
We call such a $B$ \emph{qualified}.
Note that Bob can check if a set $B$ is qualified, and clearly, the correct set of inserted keys must be qualified.
Moreover, the key property we will prove is that, in expectation, only very few $B$ are qualified.
Thus, it suffices for Alice to send few bits to tell Bob which $B$ is the correct $\insaseti[1]$ \emph{among the qualified sets}.

To prove this key property, consider the correct set of insertions $A=\insaseti[1]$ and another qualified set $B$.
Let $S_A$ [resp$.$ $S_B$] be the set of cells probed when $\insaseti[1]=A$ [resp$.$ $\insaseti[1]=B$], and let $C_A$ [resp$.$ $C_B$] be the memory state after processing the first segment.
Since $A$ and $B$ are both qualified, the contents of $S_A\cap S_B$ in $C_A$ and $C_B$ must be the same, because in particular, they are the same as in $\staend$.
For such $A$ and $B$, we say they are \emph{consistent}.
The following technical lemma says that two random subsets $A,B\subset \insaset$ of size $m/\lambda$ with intersection at most $m/2\lambda$ are consistent with very low probability.
\begin{lemma*}
	In expectation, two random $m/\lambda$-subsets $A$ and $B$ of $\insaset$ conditioned on $\left|A\cap B\right|\leq m/2\lambda$ are consistent with probability $U^{-\Omega(m/\lambda)}$.
\end{lemma*}
This is a simplified version of \cref{lm:inner} to fit in the regime in this overview.
Note that our distribution guarantees that $\insaseti[1]$ is a random subset of $\insaset$ of size $m/\lambda$.

In particular, the above lemma implies that, on average, the sets that have a small intersection with the correct $\insaseti[1]$ are very unlikely to qualify.
On the other hand, there are only $\binom{m}{\leq m/2\lambda}\binom{m/\lambda}{\geq m/2\lambda}=O(\lambda)^{m/2\lambda}$ sets that have intersection size at least $m/2\lambda$ with $A$, while there are $\approx O(\lambda)^{m/\lambda}$ many $m/\lambda$-subsets of $\insaset$.
Thus, by encoding the correct $\insaseti[1]$ among the qualified sets, Alice can save a factor of two in the message length.
This constant-factor saving in the communication cost turns out to already be sufficient for a contradiction.

\bigskip

The remaining step is to prove the above lemma, showing that two random sets conditioned on their intersection size being small are very unlikely to be consistent.
For simplicity, let us consider two random sets $A$ and $B$ conditioned on $A\cap B=\emptyset$.
If the contents of $S_A\cap S_B$ are the same after inserting either $A$ or $B$, then intuitively, these cells should not be storing much useful information about $A$ or $B$ -- we are wasting space proportional to $\left|S_A\cap S_B\right|\cdot w$.
On the other hand, the \emph{same} set $D$ is deleted in the segment in both cases, one must probe the cells storing the information about $D$.
This is because intuitively, the redundancy is only $O(n)$ bits, which is much smaller than the entropy of $D$ (or $A$, $B$) when $\ell\geq 2$;
the data structure must ``empty most of the space storing $D$'' to make room for $A$ or $B$.
Thus, the set $S_A\cap S_B$ should be sufficiently large to contain $D$.
Combining the two, we are wasting $\Omega(H(D))=\Omega(m/\lambda\cdot \log U)\gg \Omega(n)$ bits.
This leads to a contradiction (this is also where we need $\ell\geq 2$).

We formally prove the lemma by studying another communication game, which we call the \emph{inner game}.
We show that whenever two sets $A$ and $B$ are consistent, it is possible to encode sets $K,A,B,D$ (recall that $K$ is the initial set of keys) using roughly $H(K,A,B,D)-\Omega(m/\lambda \cdot \log U)$ bits.
This implies an upper bound on the possible number of such tuples $(K,A,B,D)$, i.e., only $U^{-\Omega(m/\lambda)}$-fraction of $A$ and $B$ are consistent on average.
See~\cref{sec:inner_game_protocol} for the protocol.

\bigskip

For the general case where $A$ and $B$ can have a small intersection, and at most $0.01m$ cells can be probed in more than one segment, we generalize the definition of consistency to require that the number of different cells in $S_A\cap S_B$ is bounded by $O(m/\lambda)$.
It turns out that the above argument still goes through, and Bob can recover $\insaseti[1]$ after Alice sends few extra bits.
Then, Bob simulates the operations in the first segment, and the players apply the same strategy for Bob to iteratively recover $\insaseti[2],\ldots,\insaseti[\lambda]$.
Bob learns the partition of $\insaset$ into $\insaseti[1],\ldots,\insaseti[\lambda]$, which has $\approx m\log\lambda$ bits of information, while Alice's message turns out to have only $\approx (1-\Omega(1))m\log\lambda+O(n)$ bits, yielding a contradiction.
See~\cref{sec:outer_game_protocol} and~\cref{sec:outer_game_cost} for the detailed analysis.

\section{Hard Distribution and Proof of Main Theorem}
\label{sec:hard_instance}

In the following sections, we present the formal proof of our lower bound for the dynamic succinct dictionary problem. 
Let $\sta$ be a {dynamic dictionary data structure} that maintains a set $T$ of at most $n$ distinct keys in range $[U]$, i.e., $\sta$ supports the following operations:
\begin{itemize}
\item Initialize $T$ to $\emptyset$.
\item Insert a key $x$, assuming $|T|<n$ and $x\notin T$.
\item Delete a key $x$, assuming $x\in T$.
\item Query whether $x$ is in $T$.
\end{itemize}
We will focus on the following regime:
\begin{itemize}
\item 
We assume $U=n^{1+\Theta(1)}$. Without loss of generality, we set the word size $w=\log U$ so that a key can be stored in one word.
\item The memory of $\sta$ is fixed to be $\log\binom{U}{n}+\Redun$ bits, where the term $\log\binom{U}{n}$ is the optimal space required to store $n$ keys, and $\Redun$ is the redundancy. Moreover, by the convention of prior work, we say the data structure incurs $\redun = \frac{\Redun}{n}$ \emph{wasted bits per key}.
\end{itemize}

The running time of the dictionary is measured by the expected amortized number of cell-probes per operation, while the redundancy is measured by $r$, the number of wasted bits per key.
As mentioned in \cref{sec:overview}, we will consider a sequence of $n$ insertions followed by $n$ \emph{meta-operations}.
Each meta-operation will consist of a query (for technical reasons), a deletion, followed by an insertion.
A formal description of the hard distribution is presented in Distribution~\ref{alg_hard_dist}.

\begin{algorithm}[H]
  \captionof{Distribution}{Hard Distribution}
  \label{alg_hard_dist}
  \DontPrintSemicolon
  \SetKwFunction{Query}{Query}
  \SetKwFunction{Insert}{Insert}
  \SetKwFunction{Delete}{Delete}
  Initialize an empty dictionary with capacity $n$ and key-universe $[U]$\;
  $\keyset \gets$ uniform random $n$-element subset of $[U]$\;
  Insert keys in $\keyset$ into the dictionary one by one, using $n$ insertions\;
  \For{$i = 1$ to $n$} {
    $\deli[i] \gets$ a uniform random key in $\keyset$ that has not been removed\label{step_loop_begin}\;
    \Query{$\deli[i]$}\;
    \Delete{$\deli[i]$} from the dictionary\;
    $\insai[i] \gets$ a uniform random key in $[U]$ which is neither in $\keyset$ nor in the current dictionary\;
    \Insert{$\insai[i]$} to the dictionary\label{step_loop_end}\;
  }
\end{algorithm}

The following theorem, a rephrase of \cref{thm_main} with respect to \cref{alg_hard_dist}, is the main result of this section, proving a space-time lower bound for dynamic dictionaries.

\begin{theorem}
  \label{thm:main}
  For sufficiently large $n$ and any positive integer $k \le \log^{*} n$, if a dictionary incurs $\redun \le \log^{(k)} n$ \emph{wasted bits per key}, then it must perform $\Omega(nk)$ cell-probes to process a random sequence of operations sampled from \cref{alg_hard_dist} in expectation, i.e., the expected amortized time to process one operation is $\Omega(k)$.
\end{theorem}

We first fix a data structure using $r\leq \log^{(k)} n$ wasted bits per key.
By Yao's Minimax Principle, we may assume without loss of generality that it is deterministic.

Recall that we denote Step~\ref{step_loop_begin}-\ref{step_loop_end} in the hard distribution by a meta-operation.
To prove the lower bound, we will build a tree over the sequence of $n$ meta-operations.
Each meta-operation is a leaf of the tree, which is also called a \emph{level-0} node.
The tree is then constructed bottom-up.
After constructing level $\lv-1$, we group every $\br$ \emph{consecutive} level-$(\lv - 1)$ nodes, and assign a level-$\lv$ node to be the parent of every group. 
Thus, a level-$\lv$ node $u$ represents an interval of consecutive meta-operations, which is the union of the intervals of $u$'s children.

Let $\widl$ denote the number of consecutive meta-operations a level-$\lv$ node represents.
Then, we have $\widl = \widn \cdot \br$. We will set $\widl \defeq cn \log^{(k)} n / \log^{(\lv)} n$ for levels $\lv \geq 1$, where $c = 10^6$ is a large constant, and set $\widof[0] = 1$ for the leaf nodes. For simplicity, we assume $\widn \mid \widl$.
Therefore, for $\lv > 1$, we have $\br = \log^{(\lv - 1)} n / \log^{(\lv)} n$.

The process of grouping level-$(\lv - 1)$ nodes to form level-$\lv$ nodes terminates at level $h$ when $cn \log^{(k)} n / \log^{(h)} n \ge n$, i.e., all meta-operations are grouped into the same level-$h$ node. Formally, $h = \min \{ \lv \in \mathbb{N} : \log^{(\lv)} n \le c \log^{(k)} n \}$. Here we let $\widof[h] = n$, and the only level-$h$ node is the root of the tree. 
Clearly, the height of the tree $h$ is $\Omega(k)$.
See \cref{fig:tree_structure} in \cref{sec:overview}.

\bigskip

Now let us return to the data structure.
In the cell-probe model, its memory contains $N$ cells, where $N = \left(\log\binom{U}{n} + \Redun\right) \Big/ w = O(n)$. 
The addresses of these cells are from $1$ to $N$.
We will lower bound the total number of cell-probes in two ways.
To this end, we introduce two quantities for each node $u$ in the tree, $\costu$ and $\probeu$, to measure the number of probes.

For $\costu$, if cell $i$ is probed in meta-operations $t_1$ and $t_2$ ($t_1 < t_2$) but not between them, we \emph{assign} the probe of $i$ at $t_2$ to node $u$, the \emph{lowest common ancestor (LCA)} of $t_1$ and $t_2$ in the tree, and let $\costu$ be the total number of probes assigned to $u$.
Clearly, since each cell-probe is assigned at most once, the total number of cell-probes during the whole operation sequence is at least the sum of $\costu$ over all the nodes.

For $\probeu$, we simply define it as the number of probes incurred while processing the operations in $u$'s corresponding interval. For a single level $\lv$, the sum of $\probeu$ over all level-$\lv$ nodes $u$ exactly equals the total time cost in the whole process (unlike $\costu$, we cannot use the sum of $\probeu$ over all nodes as a lower bound, because each cell-probe will be counted $h+1$ times).

The lemma below gives a lower bound on the expectations of $\costu$ and $\probeu$.

\begin{lemma}[Outer Lemma]
  \label{lm:outer}
  Let $t$ be a parameter. Suppose $U=n^{1+\alpha}$ for a constant $\alpha$, and let $\gamma=\frac{\alpha}{6(1+\alpha)}$.
  For a level-$\lv$ node $u$, if
  \begin{itemize}
  \item $2^{64}\le\br\le \frac{\alpha}{12}\log n$,
  \item $\widl\log \br\ge 100 \Redun$,
  \item $\frac{1}{2}\gamma\widl \log U \ge \Redun + 5\widl + 2(\beta + 1)\widl t$,
  \item $\widl\ge n^{1 - \gamma / 2}$,
  \end{itemize}
  where $\beta$ is a global fixed constant to be determined later, then \emph{at least one} of the two proposition holds:
  \begin{itemize}
  \item Prop 1. $\E[\costu]\ge \frac{\gamma}{100}\widl$;
  \item Prop 2. $\E[\probeu]\ge \frac{1}{16}\widl t$.
  \end{itemize}
\end{lemma}

Note that although we have set the tree parameters $\br[\lv']$, $\widof[\lv']$ above for all levels $\lv'$, \cref{lm:outer} only focuses on a single level $\lv$ and has no constraint for the tree parameters for other levels $\lv' \ne \lv$. 
Hence, even if we use another set of tree parameters, \cref{lm:outer} still applies as long as the required conditions are met (which will be the case in \cref{sec:extend}).

We will prove \Cref{lm:outer} in the next section. Now we use it to prove \Cref{thm:main}.

\begin{proofof}{\Cref{thm:main}}
The main idea is to apply \cref{lm:outer} on all nodes in levels $2 \le \lv < h$, with parameter $t=\frac{\gamma}{8(\beta + 1)}\log U$. After applying the lemma, consider two cases:
\begin{itemize}
\item If there is a level $\lv$ ($2\le \lv <h$) such that at least half of the nodes in this level satisfy \emph{Prop 2}, the expected time can be bounded by
  \[\sum_{u \textup{ in level } \lv} \E[\probeu] \ge \frac{1}{2} \cdot \frac{n}{\widl} \cdot \frac{1}{16}\widl t = \frac{1}{32} nt = \frac{\gamma}{256(\beta + 1)} \cdot n\log U.\]
  As the coefficient $\frac{\gamma}{256(\beta + 1)}$ is a constant and $\log U \ge \log n\ge \log^{*} n \ge k$, the time per operation is at least $\Omega(k)$.
\item Otherwise, it means for every level $\lv$ ($2\le \lv <h$), at least half of the nodes in it satisfy \emph{Prop 1}. Taking summation of the lower bound of $\E[\costu]$ for each node $u$ in these levels, we know that the total time cost is at least $\sum_{\lv = 2}^{h - 1} {\frac{\gamma}{100} \widl} \cdot \frac{1}{2} \cdot \frac{n}{\widl} = \frac{\gamma}{200} n (h - 2)$ in expectation, which means the expected time cost per operation is at least $\Omega(h) = \Omega(k)$, as desired.
\end{itemize}

It remains to verify the premises of \cref{lm:outer}. Notice that we require $2 \le \lv < h$, i.e., when we apply \cref{lm:outer} on some node $u$, $u$ and its children are not root or leaf nodes. Thus we always have $\widl = cn \log^{(k)} n / \log^{(\lv)} n$ and $\br = \log^{(\lv - 1)} n / \log^{(\lv)} n$.
\begin{itemize}
\item $\br\ge 2^{64}$. According to the definition of $h$, for every $2 \le \lv < h$, we have $\log^{(\lv)} n > c \log^{(k)} n \ge c = 10^6$, so $\br = \log^{(\lv - 1)}n / \log ^{(\lv)} n = 2^{\log^{(\lv)}n}/\log^{(\lv)}n \ge 2^{c} / c \gg 2^{64}$.
\item $\br \le \frac{\alpha}{12}\log n$. As $n$ is sufficiently large, we only need to show $\br = o(\log n)$, that is, $\log^{(\lv - 1)}n / \log^{(\lv)}n = o(\log n)$. All levels except for level 1 satisfy this condition.
\item $\widl \log \br \ge 100 \Redun$. We can calculate
  \begin{align*}
    \widl \log \br
    &= \widl (\log^{(\lv)}n - \log^{(\lv + 1)}n)
      \ge \frac{1}{2} \widl\log^{(\lv)}n
      = \frac{1}{2} cn\log^{(k)}n
      \gg 100n\log^{(k)}n, \\
    100\Redun &= 100n\redun
                \le 100n \log^{(k)}n,
  \end{align*}
  where the first inequality holds because $\log^{(\lv)} n \ge 2 \log^{(\lv + 1)} n$ as $\log^{(\lv)} n \ge c$. So we have $\widl \log \br \ge 100\Redun$ for all $2 \le \lv < h$.
\item $\frac{1}{2} \gamma \widl \log U \ge \Redun + 5\widl + 2(\beta + 1) \widl t$. Substituting $t=\frac{\gamma}{8(\beta + 1)} \log U$ into this inequality, we only need to prove $\frac{1}{4} \gamma \widl \log U \ge \Redun + 5\widl$. When $n$ is sufficiently large, we have
  \begin{align*}
    \gamma\widl\log U
    \ge \gamma \frac{cn\log^{(k)}n}{\log \log n}\log U
    \ge \gamma \frac{\log n}{\log \log n} \cdot cn \log^{(k)} n
    \ge 24 c n \log^{(k)} n \ge 4(\Redun + 5\widl).
  \end{align*}
\item $\widl\ge n^{1-\gamma/2}$. As $\widl\ge cn\cdot \frac{\log^{(k)}n}{\log \log n}$, this inequality naturally holds with sufficiently large $n$.
\end{itemize}
This proves the theorem.
\end{proofof}

\section{Outer Lemma}
\label{sec:outer}

Suppose $u$ is a level-$\lv$ node in the tree structure. It also represents an interval of $\widl$ consecutive meta-operations. In this section, we focus on the operations in $u$ and prove \cref{lm:outer}, which asserts that either $\E[\costu]\ge \frac{\gamma}{100}\widl$ or $\E[\probeu]\ge \frac{1}{16}\widl t$.

\subsection{Proof Overview}

We index the meta-operations in $u$ by $\{1, \ldots, \widl\}$. Let $\insai$, $\deli$ represent the keys to be inserted and deleted in the $i$-th meta-operation in $u$, respectively. Also define sequences $\insas \defeq (\insai[1], \ldots, \insai[\widl])$, $\dels \defeq (\deli[1], \ldots, \deli[\widl])$. These meta-operations are further divided into $\br$ subintervals of length $\widn$, each corresponding to a child node of $u$. We call each subinterval a \emph{segment}. 
For each $i \in [\lambda_\l]$, the keys inserted and deleted in the $i$-th segment are denoted by $\insasi \defeq \bigbk{\insai[(i - 1) \widn + 1], \ldots, \insai[i \widn]}$ and $\delsi \defeq \bigbk{\deli[(i - 1) \widn + 1], \ldots, \deli[i \widn]}$. We further define the following quantities of $\insas$:
\begin{itemize}
\item $\insaset \defeq \BK{\insai[1], \ldots, \insai[\widl]}$ is the unordered version of $\insas$, and $\insaseti[i] \defeq \bigBK{\insai[(i - 1) \widn + 1], \ldots, \insai[i \widn]}$ is the unordered version of $\insasi[i]$.
\item Inside the $i$-th segment, we use a permutation $\insorderi$ to represent the order of key insertions in $\insaseti$. Formally, $\insorderi$ is a permutation over $[\widn]$ such that the $j$-th key inserted in the $i$-th segment is the \emph{$\insorderij{j}$-th largest} element in $\insaseti$. Denote by $\insorders = \bk{\insorderi[1], \ldots, \insorderi[\br]}$ the sequence of all these permutations.
\end{itemize}

Let $\sta$ be a (memory) state of the deterministic data structure at a given time, i.e., the contents in all $N$ memory cells. It is clear that $H(\sta) \le \log\binom{U}{n} + \Redun$. We also define $\sta(\set)$ to be the state of the cells $\set$, where $\set \subseteq [N]$. Formally, $\sta(\set) \defeq \BK{(i, \, \textup{cont}(i)) : i \in \set}$, where $\textup{cont}(i)$ is the $w$-bit content of the $i$-th cell in the state $\sta$. Notice that $\sta(\set)$ also includes the address information $S$ when $S$ is a random variable. Hence, the entropy of $\sta(\set)$ might be larger than $w \cdot |\set|$.

Denote by $\stabegin$ the state of the data structure just before executing the first operation in $u$; denote by $\staend$ the state right after executing the last operation in $u$. We will not consider states before $\stabegin$ or after $\staend$ in this section.

We will prove \cref{lm:outer} by contradiction via a communication game. In this game, 
\begin{itemize}
  \item Alice is given $\insaparts$, which determines $\insaset$, and is also given $\stabegin, \dels, \insorders$;
  \item Bob is given $\insaset, \stabegin, \dels, \insorders$;
  \item the goal of the game is for Alice to tell Bob $\bk{\insaparts}$ by sending one message.
\end{itemize}
Note that $\bk{\insaparts}$ is simply a random partition of $\insaset$, which has $\widl$ elements, into $\br$ groups of size $\widn$ each. 
Since all other inputs are independent of the partition, the entropy of $\bk{\insaparts}$ given Bob's input is
\begin{align*}
  H(\insaparts \mid \insaset, \stabegin, \dels, \insorders)
  &= H(\insaparts \mid \insaset) \\
  &= \log \binom{\widl}{\widn, \ldots, \widn}
	= \log \widl ! - \br \log \widn !.
	\numberthis \label{eq:ent_partition}
\end{align*}

\subsection{Outer Game Protocol}\label{sec:outer_game_protocol}

Let us now assume for contradiction that the conclusion of \cref{lm:outer} does not hold, i.e., $\E[\costu] < \frac{\gamma \widl}{100}$ and $\E[\probeu] < \frac{1}{16} \widl t$, we will show a protocol where Alice sends less than $(\log \widl ! - \br \log \widn !)$ bits of information, but Bob can still recover $\insaparts$, which leads to a contradiction.

Recall the definition of $\costu$: While a cell $j$ is probed at different segments in $u$, we increase $\costu$ by one. Conversely, since $\costu$ is small, we know that most cells are only probed in at most one segment. We call the cells that are probed in at least two different segments \emph{bad cells}, and use $\setbad$ to denote the set of bad cells. The expected number of bad cells is at most $\E[|\setbad|] \le \E[\costu] < \frac{\gamma \widl}{100}$.

The framework of our protocol is shown in Protocol~\ref{ptc:outer}. Below, we will explain it in detail.

\begin{algorithm}[H]
  \captionof{Protocol}{Outer Game Protocol Framework}
  \label{ptc:outer}
  \DontPrintSemicolon
  Alice sends $\staend$ to Bob\;
  \For{$i = 1$ to $\br$} {
    Bob computes $\stabef$, the state of the data structure before the $i$-th segment \label{line:prepare}\;
    $\insrest \defeq \insaset \setminus (\insaseti[1] \cup \cdots \cup \insaseti[i - 1])$\;
    $\guessall \defeq \{\insetb\subseteq\insrest : |\insetb|=\widn\}$ \label{line:guessall}\;
    Alice decides $\indbad \in \{0,1\}$ and sends it to Bob \label{line:alice} \Comment{Explained below}\;
    \uIf{$\indbad = 1$} {
      Bob runs \emph{tests} on all $\insetb \in \guessall$ \Comment{Explained below}\;
      $\guesspass \defeq \{\insetb\in\guessall : \insetb \textup{ passes the test}\}$\;
      Alice sends the index of $\insaseti$ in $\guesspass$ \quad\, \Comment{It is guaranteed that $\insaseti$ will pass the test}\;
    } \Else {
      Alice sends the index of $\insaseti$ in $\guessall$\;\label{line:X=0}
    }
  }
\end{algorithm}

\paragraph{Sending \texorpdfstring{$\staend$}{C\_end}.} The first step is to send $\staend$, the final state of the data structure after processing all operations in $u$, to Bob. 
We show that the entropy of this message given Bob's knowledge is in fact much smaller than the number of bits in the memory, because $\staend$ has a large mutual information with $\stabegin$.
To upper bound this entropy, we consider the key set stored in $\staend$, which we denote by $\keyend$.
Now observe that $\keyend$ can be inferred from both $\staend$ and from $(\stabegin, \insaset, \dels)$ (by running all possible queries on $\staend$, or on $\stabegin$ and replacing keys $\dels$ with $\insaset$).
Thus, the mutual information $\mutual{\staend}{\stabegin, \insaset, \dels, \insorders} \ge H(\keyend) = \log \binom{U}{n}$. Furthermore, we have
\begin{align}
  \label{eq:ent_stateend}
  H(\staend\mid \stabegin,\insaset,\dels,\insorders) = H(\staend) - \mutual{\staend}{\stabegin,\insaset,\dels,\insorders} \le \log \binom{U}{n}+\Redun-\log\binom{U}{n} = \Redun.
\end{align}
Hence, Alice can send $\staend$ \emph{conditioned on}
 $\bk{\stabegin, \insaset, \dels, \insorders}$, using at most $\Redun$ bits in expectation.

The remainder of the protocol consists of $\br$ rounds. In each round, Bob aims to learn a single $\insaseti$.
Notice that since Bob knows the order of insertions $\insorderi$ and the sequence of deleted keys $\delsi$, Bob would also be able to infer the entire operation sequence in segment $i$.

\paragraph{Preparation steps.} At the beginning of the $i$-th round, Bob already knows $\insaseti[1], \ldots, \insaseti[i - 1]$; combined with the common knowledge $\insorders$, $\dels$, Bob can recover all operations in the first $i - 1$ segments. So he can simulate the process of the data structure from $\stabegin$ throughout the first $i - 1$ segments, obtaining the state before the $i$-th segment, $\stabef$. At this point, the set of all possible keys that can be inserted in segment $i$ \emph{in Bob's view} is $\insrest \defeq \insaset \setminus (\insaseti[1] \cup \cdots \cup \insaseti[i - 1])$. \cref{line:guessall} defines $\guessall$ as the collection of all $\widn$-element subsets of $\insrest$. Every set in $\guessall$ is currently a candidate of $\insaseti$. It is clear that the correct $\insaseti$ is in $\guessall$.

The most trivial way for Alice to reveal $\insaseti$ to Bob is to send the index of $\insaseti$ in $\guessall$. This takes $\log \bigabs{\guessall} = \log \binom{|\insrest|}{\widn}$ bits. However, this approach would not lead to a contradiction, as $\sum_{i=1}^{\br}\log \bigabs{\guessall} = \log \binom{\widl}{\widn, \ldots, \widn}$ exactly equals $\log \widl ! - \br \log \widn !$, the entropy that Bob needs to learn.

Compared to this simple method, our protocol lets Bob eliminate most of the candidates by running \emph{tests} based on $\stabef$ and $\staend$. Ideally, he can narrow down the candidates to a small subset $\guesspass \subseteq \guessall$ while the correct $\insaseti \in \guesspass$ is not eliminated. In this case, Alice only needs to send $\log \bigabs{\guesspass}\ll\log \bigabs{\guessall}$ bits. On the other hand, we observe that occasionally, $\insaseti$ is eliminated by Bob's tests.
In this \emph{bad case}, Alice will still send $\insaseti$ using the index in $\guessall$. An additional bit $\indbad$ is sent to tell Bob whether his tests will be successful.

\paragraph{Alice's Decision.} We define $\setai$ to be the set of cells probed in the $i$-th segment. Recall that we use $\setbad$ to represent the cells probed in multiple segments (Alice knows all operations, and thus, knows $\setbad$). $\indbad$ in \cref{line:alice} is determined according to the relationship between $\setai$ and $\setbad$: If $\bigabs{\setai \cap \setbad} > \gamma \widn / 2$ (recall $\gamma \defeq \frac{\alpha}{6(1 + \alpha)}$), we call it the \emph{bad case}, and Alice sends $\indbad = 0$ and the index of $\insaseti$ in $\guessall$ in \cref{line:X=0}. Otherwise, Alice sends $\indbad = 1$.

\paragraph{Bob's Tests.} For each candidate $B \in \guessall$, Bob runs the following test. He pretends that $B$ is the correct set $\insaseti$, and permutes $B$ according to the permutation $\insorderi$, obtaining an ordered sequence of $\widn$ keys $\insbsi[i]$. Then, he simulates the operations in the $i$-th segment according to $\insbsi[i]$ and $\delsi[i]$. The simulation starts with the state $\stabef$, and we denote by $\stab$ the end state after the simulation. Let $\setb$ be the set of probed cells during the simulation.

$B$ passes the test if most of these probed cells have the same content as their final value in $\staend$, i.e., if $\stab(\setb)$ and $\staend(\setb)$ differ by at most $\gamma \widn / 2$ cells. In this case, we say that $B$ is \emph{qualified}.

\bigskip

Next we show that $\insaseti$ is always qualified when $X_i=1$. Let $\staai$ be the true state after the $i$-th segment. Note that $\staai(\setai)$ and $\staend(\setai)$ can only differ in $\setai \cap \setbad$: If some cell $j \in \setai$ has different values in $\staai$ and $\staend$, it must be probed in a later segment; but it is also probed in segment $i$, which tells us that $j \in \setbad$ by the definition of $\setbad$. Combined with the fact $\bigabs{\setai \cap \setbad} \le \gamma \widn / 2$ (since $\indbad = 1$), we see that $\insaseti$ must be qualified.

Thus, in either case ($X_i=0$ or $X_i=1$), Bob is able to recover the set $\insaseti$.
By iterating over all $i=1,\ldots,\br$, Bob successfully recovers the partition $(\insaseti[1], \ldots, \insaseti[\br])$ as required.

\paragraph{Analysis.} Denote by $\msgout$ the message sent by Alice in the $i$-th round. $\msgout$ contains an indicator $\indbad$ and an index of $\log \bigabs{\guessall}$ bits (when $\indbad = 0$) or $\log \bigabs{\guesspass}$ bits (when $\indbad = 1$). Formally,
\[ |\msgout| = \indbad \log \bigabs{\guesspass} + (1 - \indbad) \log \bigabs{\guessall} + 1. \]
Thus, by~\eqref{eq:ent_stateend}, the message sent by Alice has at most
\[
  H(\staend \mid \insaset, \stabegin, \dels, \insorders)
  +
  \sum_{i=1}^{\br} \E[|\msgout|]
  \;\leq\;
  R+\sum_{i=1}^{\br} \E[|\msgout|]
\]
bits of information in total.
On the other hand, since Bob can always recover the partition $(\insaseti[1], \ldots, \insaseti[\br])$ given his input $(\insaset, \stabegin, \dels, \insorders)$ and Alice's message, we must have
\begin{align*}
  R+\sum_{i=1}^{\br} \E[|\msgout|]
  & \;\ge\;
    H(\insaseti[1], \ldots, \insaseti[\br] \mid \insaset, \stabegin, \dels, \insorders), \\
  \Leftrightarrow \qquad
  \sum_{i=1}^{\br} \E[|\msgout|]
  & \;\ge\; \log \widl ! - \br \log \widn ! - \Redun.
    \numberthis \label{eq:req_msg}
\end{align*}
Next, we will derive an upper bound on $\E[|\msgout|]$ that will contradict with this inequality when $\E[\costu] < \frac{\gamma \widl}{100}$ and $\E[\probeu] < \frac{1}{16} \widl t$.

\subsection{Estimate the Message Length}\label{sec:outer_game_cost}

We introduce an intermediate set $\guessnbr \defeq \{\insetb \in\guessall : |\insaseti \cap \insetb| \ge \widn / 2\}$ to help us estimate $\guesspass$.
This is the set of all possible $B$ that are not too different from the correct $\insaseti$.
Define indicator variables $\indnbr \defeq \indicatorEx{\guesspass \subseteq \guessnbr}$ and $\indprobe \defeq \indicatorEx{\E\!\bigBk{\bigabs{\setai}} \le \widn t}$. 
The variable $\indnbr$ indicates that the only sets that can pass the test are the ones similar to the correct $\insaseti$.
Note that $\indnbr$ is a random variable, while $\indprobe$ is not (since the only random variable in the definition of $\indprobe$ is placed within the expection), which can be fully determined by the location of the current segment (i.e., determined by the node $u$'s location in the tree and the segment index $i$). 
Later, we will demonstrate that segments with $\indprobe=1$ are \emph{likely} to have $\indnbr=1$.

Then, when $\indnbr=\indprobe=1$, we can apply $\bigabs{\guesspass}\le \bigabs{\guessnbr}$ on the inequality; otherwise, we just use the trivial bound $\bigabs{\guesspass}\le \bigabs{\guessall}$. Recall $\indbad = \indicatorEx{\bigabs{\setai \cap \setbad} \le \gamma \widn / 2}$ indicates if this is the bad case, we then have
\begin{equation}
  \label{eq:exp_msg}
  \begin{aligned}
    |\msgout|&\le(1-\indbad\indnbr\indprobe)\log \bigabs{\guessall}+\indbad\indnbr\indprobe\log \bigabs{\guessnbr}+1\\
             &=\log \bigabs{\guessall} -\indbad\indnbr\indprobe\left(\log\bigabs{\guessall}-\log \bigabs{\guessnbr}\right)+1.
  \end{aligned}
\end{equation}
The following claims give lower bounds on $\indbad$, $\indnbr$, $\indprobe$ and $\log\bigabs{\guessall}-\log\bigabs{\guessnbr}$ respectively.

\begin{claim}
  \label{clm:bound_indbad}
  We have
  \[\E\left[\sum_{i=1}^{\br / 2}\indbad\right]\ge \frac{1}{4}\br.\]
\end{claim}
\begin{proof}
  First, we reconsider how $\setai\cap\setbad$ in the $i$-th segment is related to $\costu$. If some cell is probed in $s>1$ different segments of $u$, it incurs $s - 1 \ge s / 2$ costs on $u$. Hence, if we think every probe of a bad cell incurs $1/2$ cost, we will get a lower bound on $\costu$. Noting that $\setai \cap \setbad$ is exactly the set of bad cells probed in segment $i$, where each probe contributes at least $1/2$ to $\costu$, we have
  \begin{equation}
    \label{eq:costandx}
    \frac{1}{2}\E\Bk{\sum_{i=1}^{\br}\bigabs{\setai \cap \setbad}}\le \E[\costu]\le \frac{\gamma}{100}\widl.
  \end{equation}
  Then we put $\indbad$ into \eqref{eq:costandx}. Since $\indbad=0$ implies $\bigabs{\setai \cap \setbad}>\gamma\widn/2$,
  \begin{gather*}
    \frac{1}{2} \E\Bk{\sum_{i=1}^{\br} \frac{\gamma \widn}{2} (1 - \indbad)}
    \le \frac{1}{2} \E\Bk{\sum_{i=1}^{\br} |\setai \cap \setbad|}
    \le \frac{\gamma}{100} \widl =  \frac{\gamma \widn \cdot \br}{100} \\
    \quad \Rightarrow \quad
    \E\Bk{\sum_{i=1}^{\br}(1-\indbad)} \le \frac{\br}{25}, \\
    \E\left[\sum_{i=1}^{\br/2}\indbad\right]
    \ge \frac{\br}{2} - \E\left[\sum_{i=1}^{\br}(1-\indbad)\right]
    \ge \frac{\br}{2} - \frac{\br}{25} \ge \frac{\br}{4}.
    \qedhere
  \end{gather*}
\end{proof}

The next claim gives a bound on $\indnbr$.
\begin{claim}
  \label{clm:bound_indnbr}
  For any $1 \le i \le \br / 2$, if $\indprobe = 1$, we have
  \[\Pr[\indnbr = 0 \land \indbad = 1] \le \frac{1}{8}.\]
\end{claim}
\begin{proof}
  Suppose $\indbad = 1$. For any $\insetb \in \guesspass$, we know $\insaseti$ and $\insetb$ are both qualified, which means $\staai(\setai)$ and $\staend(\setai)$ differ in at most $\gamma\widn/2$ cells; so do $\stab(\setb)$ and $\staend(\setb)$. Thus $\staai(\setai\cap\setb)$ and $\stab(\setai\cap\setb)$ have at most $\gamma\widn$ different cells, in which case we say $\insaseti$ and $\insetb$ are \emph{consistent}. Formally:

  \begin{definition}
    \label{def:consistency}
    For two insertion sets $\inseta, \insetb \in \guessall$, we say $\inseta$ and $\insetb$ are consistent, if $\staa(\seta\cap\setb)$ and $\stab(\seta\cap\setb)$ have at most $\gamma\widn$ different cells.
  \end{definition}

  Under this definition, if $\indbad = 1$ and $\insetb \in \guesspass$, then $\insaseti$ and $\insetb$ must be consistent. Hence,
  \begin{align*}
    \Pr \Bk{\indnbr = 0 \land \indbad = 1}
    &\le \Pr \Bk{\bk{\exists \insetb\in \guessall\setminus\guessnbr, \; \insetb \in \guesspass} \land (\indbad = 1)}\\
    &\le \Pr \Bk{\exists \insetb\in \guessall\setminus\guessnbr, \; \insetb \textup{ and } \insaseti \textup{ are consistent}}.
  \end{align*}
  Thus, to prove \cref{clm:bound_indnbr}, it suffices to show
  \[
    \Pr\Bk{\exists \insetb \in \guessall \setminus \guessnbr, \; \insetb \textup{ and } \insaseti \textup{ are consistent}} \le \frac{1}{8}.
  \]
  Next, we show that a random $\insetb \notin \guessnbr$ is unlikely to be consistent with $\insaseti$.
  
  \begin{restatable}[Inner Lemma]{lemma}{InnerLemma}
    \label{lm:inner}
    Assume $\stabef$ (the state before the $i$-th segment) and $\insrest$ are random variables whose distribution is induced from Distribution~\ref{alg_hard_dist}. Let $A, B$ be two uniformly random $\widn$-sized subsets of $\insrest$ conditioned on $|A \cap B| < \widn / 2$. If for some parameter $t$, inequalities $\frac{1}{2}\gamma\widl \log U \ge \Redun + 5\widl + 2(\beta + 1)\widl t$ ($\beta$ is a global fixed constant), $\widl \ge n^{1-\gamma/2}$, and $\E[|\seta|] \le \widn t$ hold, then we have
    \[\Pr_{\stabef, \insrest, \inseta, \insetb}\left[\inseta\textup{ and }\insetb\textup{ are consistent} \;\middle|\; \vphantom{\Big|} |\inseta\cap\insetb| < \widn / 2 \right]\le U^{-\gamma\widn}.\]
  \end{restatable}

  We will prove this lemma in \cref{sec:inner}.
  Now, we continue the proof of \cref{lm:outer} using this lemma.

  To apply this lemma, observe that the first two inequalities ($\frac{1}{2}\gamma\widl \log U \ge \Redun + 5\widl + 2(\beta + 1)\widl t$ and $\widl \ge n^{1-\gamma/2}$) hold because they are also premises of \cref{lm:outer}; the third constraint ($\E[|\seta|] \le \widn t$) holds because $\indprobe = 1$.
  Moreover, note that $\insaseti$ is a random subset of $\insrest$; if we randomly select a $\insetb \in \guessall \setminus \guessnbr$, the distribution of $(\insaseti, \insetb)$ matches the required distribution of $(\inseta,\insetb)$ in the lemma.

  Combining with the discussion above, as long as $\indprobe=1$, we have
  \begin{align*}
    \Pr_{\stabef, \insrest, \insaseti}[\indnbr = 0 \land \indbad = 1]
    &\le \Pr\Bk{\exists \insetb\in \guessall\setminus\guessnbr, \; \insaseti \textup{ and } \insetb \textup{ are consistent}} \\
    &\le \E\Bk{\abs{\BK{\insetb\in\guessall\setminus \guessnbr : \insaseti \textup{ and } \insetb \textup{ are consistent}}}} \\
    &= \bigabs{\guessall \setminus \guessnbr} \cdot \Pr_{\substack{\stabef, \insrest, \insaseti \\ \insetb \in \guessall\setminus\guessnbr}}
    \Bk{\insaseti \textup{ and } \insetb \textup{ are consistent}} \\
    &\le \bigabs{\guessall} \cdot U^{-\gamma \widn} \tag{$\ast$}\label{eq:apply_inner} \\
    &\le 2^{\widl} \cdot U^{-\gamma \widn} = 2^{\br\widn} \cdot n^{-\alpha \widn/6},
      \numberthis \label{eq:upperofy}
  \end{align*}
  where Step~\eqref{eq:apply_inner} is derived by applying \cref{lm:inner} with $(A, B) = (\insaseti, B)$. (Recall $\gamma = \frac{\alpha}{6(1+\alpha)}$.)

  Recall that in the statement of \cref{lm:outer}, we require $\br\le \log n^{\alpha / 12}$. Plugging this into \eqref{eq:upperofy}, we obtain $\Pr[\indnbr = 0 \land \indbad = 1]\le n^{-\alpha\widn/12}$. Since $\alpha$ is a constant in our assumption $U=\poly(n)$, with sufficiently large $n$, we have $\Pr[\indnbr=0 \land \indbad = 1] \le 1/8$. (This probability can be smaller than any fixed constant, but $1/8$ is enough for us.)
\end{proof}

Using the condition $\E[\probeu] \le \frac{1}{16} \widl t$, we can derive a bound of $\sum \indprobe$ as below.
\begin{claim}
  \label{clm:bound_probe}
  We have
  \[\sum_{i=1}^{\br/2} (1-\indprobe) \le \frac{1}{16} \br.\]
\end{claim}
\begin{proof}
  Note that $\sum_{i=1}^{\br/2} \bigabs{\setai} \le \sum_{i=1}^{\br} \bigabs{\setai} \le \probeu$. Combined with $\indprobe \defeq \indicatorEx{\E\!\bigBk{\bigabs{\setai}} \le \widn\cdot t}$, we get
  \[
    \sum_{i=1}^{\br/2} (1-\indprobe) \cdot \widn \cdot t
    \le \sum_{i=1}^{\br / 2} \E\!\bigBk{\bigabs{\setai}}
    \le \E[\probeu] \le \frac{1}{16} \widl\cdot t
    \quad \Rightarrow \quad \sum_{i=1}^{\br/2} (1-\indprobe) \le \frac{1}{16} \br. \qedhere
  \]
\end{proof}

Next, a simple counting argument implies a lower bound on $\log \bigabs{\guessall} - \log \bigabs{\guessnbr}$.
\begin{claim}
  \label{clm:bound_guess}
  For $1 \le i \le \br / 2$, we have
  \[ \log\bigabs{\guessall} - \log \bigabs{\guessnbr} \ge \frac{1}{4} \widn \log \br. \]
\end{claim}
\begin{proof}
  By definition, $\bigabs{\guessall} = \binom{|\insrest|}{\widn}$, 
  and since $i\le \br/2$, $\insrest$ satisfies $|\insrest|\ge \widl/2$, so
  \[
    \log \bigabs{\guessall} \ge \log \binom{\widl / 2}{\widn} \ge \log \bk{\frac{\widl / 2}{\widn}}^{\widn} = \widn \log \br - \widn. \numberthis \label{eq:guessall}
  \]
  
  On the other hand, we have
  \begin{align*}
    \bigabs{\guessnbr} &= \sum_{j = \widn / 2}^{\widn} \binom{\widn}{j} \binom{|\insrest| - \widn}{\widn - j} \\
                &\le \widn \binom{\widn}{\widn / 2} \binom{|\insrest|}{\widn / 2}, \\
    \log \bigabs{\guessnbr} &\le \log \widn + \widn + \log \binom{\widl}{\widn / 2} \\
                &\le 2\widn + \frac{\widn}{2} \log 2e \br\\
                &\le 4\widn + \frac{\widn}{2} \log \br.
  \end{align*}
  On the first line above, we count the number of elements $\insetb \in \guessnbr$ by enumerating $j = |\insaseti \cap \insetb|$; the two binomial coefficient factors represent the number of ways to select $\insaseti \cap \insetb$ and $\insetb \setminus \insaseti$, respectively. The inequality on the second line holds because both binomial coefficients reach their maximum values when $j = \widn / 2$. Combining the last line with \eqref{eq:guessall}, we have
  \begin{align*}
    \log \bigabs{\guessall} - \log \bigabs{\guessnbr} &\ge \frac{\widn}{2} \log \br - 5 \widn\\
                                        &\ge \frac{1}{4} \widn \log \br
  \end{align*}
  as we have assumed $\log \br \ge 64$. This proves the claim.
\end{proof}

Now we are ready to finish the proof of \cref{lm:outer}:

\begin{proofof}{\cref{lm:outer}}
  Suppose the lemma does not hold. Recall \eqref{eq:exp_msg}:
  \[
    |\msgout| \le \log \bigabs{\guessall} -\indbad\indnbr\indprobe\left(\log\bigabs{\guessall}-\log \bigabs{\guessnbr}\right) + 1.
  \]
  For the segments with $i \le \br / 2$, we use the previous claims to bound each term in this inequality. From \cref{clm:bound_indbad,clm:bound_indnbr,clm:bound_probe}, we can bound the expectation of $\indbad\indnbr\indprobe$ by writing it as
  \begin{align*}
    \E \Bk{\indbad\indnbr\indprobe} 
    \ge{}& \E\Bk{\indbad\indnbr\indprobe - (1 - \indbad)(1 - \indprobe)}\\
    ={}& \E\Bk{\indbad - (1 - \indprobe) - \indbad \indprobe (1 - \indnbr)}\\
    ={}& \E\Bk{\indbad} -  \bk{1 - \indprobe} - \indprobe\Pr\Bk{\indnbr = 0 \land \indbad = 1}.
  \end{align*}
  \cref{clm:bound_guess} is a direct upper bound on the term $\log\bigabs{\guessall}-\log \bigabs{\guessnbr}$.

  For other segments with $i > \br/2$, we simply apply the trivial bound $|\msgout| \le \log \bigabs{\guessall} + 1$. Now taking summation over all the segments, i.e., $i = 1, \ldots, \br$, we get
  \begin{align*}
    &\sum_{i=1}^{\br}\E[|\msgout|]
    \;\le\;\br+ \sum_{i=1}^{\br}\log\bigabs{\guessall}-\sum_{i=1}^{\br/2} \Bigbk{\E\Bk{\indbad} -  \bk{1 - \indprobe} - \indprobe\Pr\Bk{\indnbr = 0 \land \indbad = 1}}\cdot \frac{1}{4}\widn \log \br \\
    &=\br+ \sum_{i=1}^{\br}\log\bigabs{\guessall}-\bk{\E\Bk{\sum_{i=1}^{\br/2}\indbad} - \sum_{i=1}^{\br/2} (1-\indprobe) - \sum_{\substack{1 \le i \le \br / 2 \\ \indprobe=1}}\Pr[\indnbr=0 \land \indbad = 1]} \cdot \frac{1}{4}\widn \log \br \\
    &\le \br+ \sum_{i=1}^{\br}\log\bigabs{\guessall} - \bk{\frac{\br}{4} - \frac{\br}{16} - \frac{1}{8} \cdot \frac{\br}{2}} \cdot \frac{1}{4}\widn \log \br  \\
    \numberthis\label{eq:expofmsg2}
    &=\br+ \log \widl! - \br \log \widn!-\frac{\widl}{32}\log\br.
  \end{align*}
  On the other side, recall \eqref{eq:req_msg}:
  \[
    \sum_{i=1}^{\br} \E\Bk{|\msgout|} \ge \log \widl ! - \br \log \widn ! - \Redun.
  \]
  Combining it with \eqref{eq:expofmsg2}, we get
  \begin{align*}
    \Redun \ge \frac{\widl}{32} \log \br - \br  > \frac{\widl}{64} \log \br,
    \numberthis \label{eq:redun_contradict}
  \end{align*}
  where the second inequality is from $\br \le \widl$ and $\log \br \ge 64$. Finally, observe that~\eqref{eq:redun_contradict} contradicts with the premise of the lemma, $R \le \frac{\widl}{100} \log \br$.
  This completes the proof of Lemma~\ref{lm:outer}.
\end{proofof}

\section{Inner Lemma}
\label{sec:inner}

In this section, we prove \cref{lm:inner}, the last piece of our lower bound.

\InnerLemma*

The view here is a bit different from previous sections as we focus on a \emph{single} segment of length $\widn$ instead of all $\br$ segments within a node $u$. 
We start by restating the procedure within the single segment using more specific notations.

\begin{enumerate}
\item Parameters. We still use $n$ to represent the number of keys stored in the data structure, and $U=n^{1+\alpha}$ to represent the size of the key universe. Define $\wid<n/2$ as the number of meta-operations in the segment. When applying this lemma in \cref{sec:outer}, we set $\wid=\widn$.
  
  In the statement of \cref{lm:inner} above, the probability is conditioned on $|A \cap B| < \wid / 2$. Below, we first prove the lemma conditioned on $|A \cap B| = g$ for every fixed $g < \wid / 2$. Then \cref{lm:inner} can be directly implied via the law of total probability.
  
  The parameters $n, U, \wid, g$ are fixed, while all other involved variables are random according to the hard distribution. Their distributions will be explicitly stated below.
\item Initial keys. Let $\keyset$ be the key set stored in $\stabef$, the state of the data structure just before the current segment. It is chosen uniformly at random among all $n$-sized subsets of the universe $[U]$. $\stabef$ is also a random state, about which we only know that its key set $K$ follows the above uniform distribution.
  
\item Deletion sequence. In \cref{sec:outer}, we used $\delsi$ to denote the keys to be deleted in the $i$-th segment. Now we omit the superscript $i$ since we are focusing on a single segment, using $\dels$ to denote the deletion sequence in the segment. We further divide the information of $\dels$ into two parts: the set of keys $\delset \defeq \BK{\deli[1], \ldots, \deli[\wid]}$ to be deleted, and the order in which the keys in $\delset$ are deleted. The latter one is formulated by a permutation $\delorder$ over $[\wid]$.
  
  Conditioned on $K$, it is easy to see that $\delset$ follows the uniform distribution over all $\wid$-element subsets of $K$ and $\delorder$ is a random permutation.
  
\item Insertion sets. The two random insertion sets $A$ and $B$ are uniformly chosen over $\wid$-element subsets of $[U] \setminus K$, conditioned on $|A \cap B| = g$. Note that in the original description, $A, B$ are sampled within $\insrest$, but $\insrest$ is a random set by itself.
Thus, we can avoid using the intermediate variable $\insrest$ and directly describe the distribution of $(A, B)$ as above. $\insrest$ will not appear throughout the proof of \cref{lm:inner}. Also note that if we hide $B$ and only observe the distribution of $A$ (conditioned on $K$), it is a uniformly random $\wid$-sized subset of $[U] \setminus K$, which exactly matches the distribution of the insertion set in the current segment (induced from the hard instance).
  
  We use a permutation $\insorder$ to represent the order in which the keys are inserted. It was denoted by $\insorderi$ in the outer game. We denote by $\insas, \insbs$ the ordered sequences of inserted keys, which are obtained by permuting $A, B$ with the same permutation $\insorder$. ($\insas$ was denoted by $\insasi$ in the outer game.)
  
\item Ending states. For any $1 \le i \le \wid$, recall that meta-operation $i$ consists of three operations: \Query{$\deli$}, \Delete{$\deli$}, \Insert{$\insai$}, assuming the insertion sequence $\insas$ is processed. The state of the data structure after executing all $\wid$ meta-operations on $\stabef$ is denoted by $\staa$. Similarly, if we replace $\insas$ with $\insbs$ and execute these operations on $\stabef$, we obtain another ending state $\stab$. Let $\seta, \setb$ be the sets of probed cells during these two procedures, respectively. Note that from the condition of \Cref{lm:inner}, $\E[|\seta|] = \E[|\setb|]\le \wid t$.
  
\item Consistency. Recall \cref{def:consistency}: We say $\inseta$ and $\insetb$ are consistent if $\staa(\seta\cap\setb)$ and $\stab(\seta\cap\setb)$ have at most $\gamma\wid$ different cells, where $\gamma \defeq \frac{\alpha}{6(1 + \alpha)}$ is a fixed constant.
\end{enumerate}

Based on these definitions, we restate \cref{lm:inner} as follows:

\begin{lemma}[Inner Lemma Restated]
  \label{lm:inner_re}
  For integers $U, n, \wid, g$ satisfying $g < \wid / 2$, $\wid < n / 2$, $U = \poly(n)$, assume the random variables $\stabef, \dels, \insorder, A, B$ are randomly sampled according to the procedure above. Let $t > 0$ be a parameter such that the constraints $\frac{1}{2}\gamma\wid\log U \ge \Redun + 5\wid + 2(\beta + 1)\wid t$ ($\beta$ is a global fixed constant), $\wid\ge n^{1-\gamma/2}$, and $\E[|\seta|]\le \wid t$ are satisfied, we must have
  \[ \Pr\!\BigBk{ \inseta \textup{ and } \insetb \textup{ are consistent} \;\Big|\; |\inseta\cap\insetb|=g} \le U^{-\gamma\wid}. \]
\end{lemma}

Next, we introduce a new communication game to prove this lemma.

\subsection{Inner Game}\label{sec:inner_game_protocol}

Sender Alice and receiver Bob are playing a communication game based on the random process above. The public information known by them are the fixed parameters $g,\wid,n,U$. In addition, Alice knows the initial state $\stabef$ and sequences $\insas,\insbs,\dels$; Bob only knows permutations $\insorder$ and $\delorder$. Now Alice wants to send Bob $A, B, D$ and the initial key set $\keyset$ when $\inseta$ and $\insetb$ are consistent. Since Bob knows the two permutations, after learning $A, B, D$, he knows the complete sequences $\insas, \insbs, \dels$. To achieve this goal, they follow the protocol introduced below, which makes use of the following data structure from \cite{chazelle2004bloomier}.

\begin{theorem}[Bloomier Filter \cite{chazelle2004bloomier}]
  \label{thm:bloomier}
  Let $A, B$ be two given disjoint subsets of $[U]$. There is a set $S$ separating $A, B$, i.e., $A \subseteq S$ and $B \cap S = \emptyset$, that we can encode within $O(|A| + |B|)$ bits in expectation. We denote the leading factor within the big-$O$ notation by $\beta$, which is a global fixed constant that appeared in \cref{lm:outer,lm:inner}. (Formally, the expected encoding length is at most $\beta(|A| + |B|)$ for all $|A| + |B| > 0$.)
\end{theorem}

The protocol is as follows:

\begin{enumerate}
\item Check consistency. First, Alice determines an indicator variable $\indin \defeq \indicator{A \textup{ and } B \textup{ are consistent}}$ and sends it to Bob. If $\inseta$ and $\insetb$ are not consistent, the game is terminated immediately. Otherwise, they continue with the subsequent steps.
  
\item Send $B$. Alice directly sends $\insetb$ using $\log\binom{U}{\wid}$ bits. Since Bob knows $\insorder$, he learns $\insbs$ after receiving the message.

\item Send cell contents. Alice sends an artificial memory state $\stamix$, which is a combination of $\stabef$ and $\staa$. For the cells in $\setb$, $\stamix\bk{\setb} = \stabef(\setb)$; for the cells outside $\setb$, $\stamix(\setnotb) = \staa(\setnotb)$. Sending this state takes $Nw = \log \binom{U}{n} + \Redun$ bits.
  
  To help Bob use this mixed state correctly, Alice also needs to tell Bob which cells in $\stamix$ come from $\stabef$ and which cells come from $\staa$. The most trivial way of sending this ``partition information'' is to send one bit for each cell to indicate where it comes from, which takes a total of $O(n)$ bits. Unfortunately, this trivial way does not suffice for all the cases we concern,\footnote{Actually, it is enough to use this $O(n)$-bit approach to prove \cref{thm:main} where we only use \cref{lm:inner_re} for $m\ge \frac{n}{\log \log n}$. The following improved approach is needed in the later discussion about an extension of \cref{thm:main} (i.e., \cref{thm:ext}), where we need to use \cref{lm:inner_re} for much smaller $m$ where an $O(n)$ cost is unaffordable.} and we use a better approach as follows.

  Notice that $\stabef(\setnotab) = \staa(\setnotab)$. For cells in $\setnotab$, it is correct to treat them as coming from either $\stabef$ or $\staa$. Hence Alice only needs to send \emph{some} partition $\bk{\setb^*; \setnotbstar}$ that agrees with the true partition $\bk{\setb; \setnotb}$ in $\seta \cup \setb$ (i.e., $\setb \subseteq \setb^*$ and $\setab \subseteq \setnotbstar$).

  Fortunately, the Bloomier Filter stated in \cref{thm:bloomier} solves this task. It can encode such a partition $\bk{\setb^*; \setnotbstar}$ using $\beta |\seta \cup \setb|$ bits in expectation for a global constant $\beta > 0$. Alice directly sends the Bloomier Filter to Bob, which takes $\beta |\seta \cup \setb| \le 2 \beta \wid t$ bits in expectation, since we have assumed $\E[|\seta|] = \E[|\setb|]\le\wid t$. Then Bob learns the partition $\bk{\setb^*; \setnotbstar}$; in particular, Bob now knows $\stabef(\setb^*)$.

\item\label{step:inner_deleted_set} Do simulation and recover $D$. Bob runs the following test for all possible choices of $\delset$ on $\stabef(\setb^*)$. He enumerates all $\wid$-sized subsets $\delset^* \subset [U]$, and permutes it according to $\delorder$ to get the deletion sequence $\dels^*$. He then performs $\wid$ meta-operations from $\stabef$ of the form ``\Query{$\deli^*$} -- \Delete{$\deli^*$} -- \Insert{$\insbi$}'', pretending that $\delset^*$ is the true deletion set. During the test, once a cell outside $\setb^*$ is probed, or any \Query returns ``not exist'', the test terminates and $\delset^*$ \emph{fails the test}. If no such exception occurs during the whole simulation, we say $\delset^*$ \emph{passes the test}.
  
  It is clear that the correct $\delset$ will pass the test. In fact, to pass the test, $\delset^*$ must satisfy $\delset^*\subseteq \keyset\cup\insetb$, as the data structure will not return ``exist'' for any key it does not store. Let $\guess=\{\delset^* : \delset^*\textup{ passes the test}\}$, then $|\guess|\le \binom{|\keyset\cup\insetb|}{\wid} = \binom{n+\wid}{\wid}$. So Alice can send Bob an index in $\guess$ within $\log \binom{n+\wid}{\wid}$ bits, to tell him the index of $\delset$ in $\guess$. After receiving the message, Bob recovers $\delset$ and thus knows $\dels$.

\item Recover $\setb$ and $\seta \cap \setb$. After knowing $\dels$ and $\insbs$, Bob can simulate the correct operation sequence again, recording which cells are probed during the process, namely $\setb$. Then, Alice sends the subset $\seta \cap \setb \subseteq \setb$ to Bob \emph{given} $\setb$. It takes at most $\E[|\setb|] \le \wid t$ bits in expectation.
  
\item Recover $\inseta$ and $\keyset$.
  Alice sends Bob the cells where $\staa(\seta\cap\setb)$ and $\stab(\seta\cap\setb)$ are different, including both the address information (i.e., the set of indices of these cells) and their cell contents. The former costs $\E\Bk{|\setb|} \le \wid t$ bits in expectation, as it is a subset of $\setb$. For the latter, since $\inseta$ and $\insetb$ are consistent, the number of such cells will not exceed $\gamma\wid$ by definition. Thus, we need totally $\wid t + \gamma \wid w = \wid t + \gamma\wid\log U$ bits to send the message (we have assumed $w=\log U$).

  After that, Bob can use this message to obtain $\staa$:
  \begin{itemize}
  \item In the cells $\seta \cap \setb$, $\staa$ and $\stab$ have limited difference, which are already sent to Bob, so Bob can edit $\stab(\seta\cap\setb)$ to obtain $\staa(\seta\cap\setb)$.
  \item The cells $\setnota$ are not probed before reaching $\staa$, so in these cells we have $\staa(\setnota) = \stabef(\setnota)=\stamix(\setnota)$.
  \item Moreover, $\staa(\setab)=\stamix(\setab)$.
  \end{itemize}
  Specifically, the last two parts imply that $\staa(\setnotacapb) = \stamix(\setnotacapb)$.
  Bob learns $\staa$ by merging it with the first part $\staa(\seta \cap \setb)$. Note that Bob does not know $\seta$ (yet); he only knows $\setb$ and $\seta \cap \setb$. But $\seta$ is not needed for this merging process.

  At this point, Bob learns the whole state $\staa$, and hence knows its key set $(K \setminus D) \cup A$. The final step is to send a $\log\binom{n}{\wid}$-bit message specifying $A$ given $(K \setminus D) \cup A$. Then Bob can infer $A$ and $K$ respectively.
\end{enumerate}

  Thus, Bob recovers the four sets $A,B,D,K$ in this whole process when $A$ and $B$ are consistent ($W=1$).
Denote the message Alice sends (including $\indin$) by $\msgfrom$, and denote the information Bob learns about $\inseta, \insetb, \delset, \keyset$ by $\msgto$. Here $\msgfrom$ and $\msgto$ can be seen as random variables about the random process. Since $\msgto$ can be inferred from $\msgfrom$ under any condition, we have
\begin{align*}
  H(\msgfrom \mid \indin = 1) \ge H(\msgto \mid \indin = 1). \numberthis \label{eq:msg_ent_cmp}
\end{align*}
In the next subsection, we will analyze the entropy of $\msgfrom$ and $\msgto$ according to each step of communication.

\subsection{Entropy Calculation}

For $\msgfrom$, we add the message sent in each step together:
\[
  \begin{aligned}
  H(\msgfrom \mid \indin=1)
	\le& \log\binom{U}{\wid} + \bk{\log\binom{U}{n} + \Redun} + 2\beta \wid t + \log\binom{n+\wid}{\wid} \\ &+ \wid t +\bk{\wid t + \gamma \wid \log U} + \log\binom{n}{\wid},
  \end{aligned}
  \numberthis\label{eq:entropy_ma}
\]
where the terms on the RHS correspond to the cost of sending $\insetb$, $\stamix$, the Bloomier Filter, the index of $\delset$, $\seta \cap \setb$ conditioned on $\setb$, the difference between $\staa(\seta \cap \setb)$ and $\stab(\seta \cap \setb)$, and $\inseta$ conditioned on $(K \setminus D) \cup A$, respectively.

All information Bob learns is $\msgto=(\inseta,\insetb,\delset,\keyset)$, so we have
\begin{align*}
	\numberthis\label{eq:entropy_mb}
	H(\msgto\mid\indin=1)
	&\ge\log\left(1 \Big/ \!
	\max_{\inseta^*,\insetb^*,\delset^*,\keyset^*} \, \Pr_{\inseta,\insetb,\delset,\keyset}
	\bigBk{
		\inseta \!=\! \inseta^*,
		\insetb \!=\! \insetb^*,
		\delset \!=\! \delset^*,
		\keyset \!=\! \keyset^*
		\mid\indin=1}
	\right).
\end{align*}
We can see that for any $\inseta^*,\insetb^*,\delset^*,\keyset^*$, the probability term can be rewritten as
\begin{align*}
	&\Pr[\inseta=\inseta^*,\insetb=\insetb^*,\delset=\delset^*,\keyset=\keyset^*\mid \indin=1]\\
	&\le\frac{\Pr[\inseta=\inseta^*,\insetb=\insetb^*,\delset=\delset^*,\keyset=\keyset^*]}{\Pr[\indin=1]}\\\numberthis\label{eq:probterm}
	&=\left(\Pr[\indin=1] \cdot \binom{U}{n}\cdot \binom{n}{\wid}\cdot \binom{U-n}{\wid}\cdot \binom{\wid}{g}\cdot \binom{U-n-\wid}{\wid-g}\right)^{-1},
\end{align*}
where the binomial coefficients in the last line represent the number of ways to choose $\keyset$, $(\delset \mid \keyset)$, $(\inseta \mid \keyset)$, $(\inseta \cap \insetb \mid \inseta)$, and $(\insetb \setminus \inseta \mid \keyset, \inseta)$, respectively.

Substituting \eqref{eq:probterm} into \eqref{eq:entropy_mb}, we get
\begin{align*}
	&H(\msgto\mid \indin=1)\\
	&\ge \log \Pr[\indin=1]+\log \binom{U}{n}+\log \binom{n}{\wid}+\log \binom{U-n}{\wid}+\log \binom{\wid}{g}+\log \binom{U-n-\wid}{\wid-g}.
\end{align*}
Again, substituting this bound and \eqref{eq:entropy_ma} into \eqref{eq:msg_ent_cmp} gives
\begin{align*}
  &\log \frac{1}{\Pr[\indin=1]} \\
  \ge \;
  & \log \binom{U}{n} + \log \binom{n}{\wid} + \log \binom{U - n}{\wid} + \log \binom{\wid}{g} + \log \binom{U-n-\wid}{\wid-g} \\
  & - \log \binom{U}{\wid} - \log \binom{U}{n} - R - 2\wid(\beta + 1)t - \log \binom{n + \wid}{\wid} - \gamma \wid \log U - \log \binom{n}{\wid} \\
  = \;
  & \log \binom{\wid}{g} + \log \binom{U - n - \wid}{\wid - g} - \bk{\log \binom{U}{\wid} - \log \binom{U - n}{m}} \\
  & - \log \binom{n + \wid}{\wid} - R - 2\wid(\beta + 1) t - \gamma \wid \log U.
    \numberthis\label{eq:final_pr}
\end{align*}
We simplify the RHS with the following facts:
\begin{itemize}
\item $\log\binom{U}{\wid}-\log\binom{U-n}{\wid} \le \wid\log\frac{U-\wid}{U-n-\wid} \le \wid$ as long as $n$ is sufficiently large.
\item Similarly, $\log \binom{U - n - \wid}{\wid - g} \ge \log \binom{U}{\wid - g} - m$.
\item $\log\binom{n + \wid}{\wid} \le \wid \log \frac{e(n + m)}{\wid} \le \wid \log \frac{2en}{m} \le \wid\bk{3 + \log \frac{n}{\wid}}$; since we required $\wid\ge n^{1-\gamma/2}$, this term is at most $\frac{\gamma}{2}\wid \log n + 3\wid \le \frac{\gamma}{2}\wid \log U + 3 \wid$.
\item Recall that we have $g < \frac{\wid}{2}$, $\wid < \frac{U}{2}$, and $\gamma=\frac{\alpha}{6(1+\alpha)}$, so
  \begin{align*}
    \log \binom{U}{\wid-g} \ge \log \binom{U}{\wid / 2}
    \ge \frac{\wid}{2}\log \frac{U}{\wid}
    \ge\frac{\wid}{2}\log\frac{U}{U^{1/(1+\alpha)}}
    = \frac{m}{2} \bk{1 - \frac{1}{1 + \alpha}} \log U
    = 3\gamma \wid\log U,
  \end{align*}
  where the third inequality holds as $m \le n = U^{1 / (1 + \alpha)}$. Therefore, $\log \binom{U - n - \wid}{\wid - g} \ge \log \binom{U}{\wid - g} - m \ge 3 \gamma \wid \log U - \wid$.
  \item $\log\binom{\wid}{g} \ge 0$, we just omit it.
\end{itemize}
Therefore, \cref{eq:final_pr} can be rewritten as
\begin{align*}
  \log \frac{1}{\Pr[\indin = 1]} &\ge (3 \gamma \wid \log U - m) - m - \bk{\frac{\gamma}{2}m \log U + 3m} - R - 2\wid (\beta + 1) t - \gamma \wid \log U \\
  &\ge \frac{3}{2} \gamma \wid \log U - (5m + R + 2\wid (\beta + 1) t).
\end{align*}
Recall that we have $R + 5m + 2\wid (\beta + 1) t \le \frac{1}{2} \gamma \wid \log U$ as a premise of \cref{lm:inner_re}, so we have $\log (1/\Pr[\indin = 1]) \ge \gamma \wid \log U$.
Thus,
\[\Pr\left[\vphantom{\sum} A, B \textup{ are consistent} \;\middle|\; |A \cap B| = g\right] = \Pr[\indin = 1] \le 2^{-\gamma \wid \log U} = U^{-\gamma \wid},\]
which concludes the proof of \cref{lm:inner_re}. \cref{lm:inner} is a direct corollary of \cref{lm:inner_re}, as we only need to apply the law of total probability over $g < \wid/2$.

\section{Extended Lower Bound}
\label{sec:extend}

In the previous sections, we have proved a space-time lower bound for dynamic dictionaries: If the data structure has redundancy at most $n \log^{(k)} n$ (i.e., incurs at most $\log^{(k)} n$ wasted bits per key), its expected amortized time complexity is at least $\Omega(k)$. In this section, we extend this result to dictionaries with sublinear redundancy $\Redun < n$. The result is stated as the following theorem, rephrasing~\cref{thm_main_2} with respect to \cref{alg_hard_dist}.
\begin{theorem}
  \label{thm:ext}
  For any dynamic dictionary with redundancy $\Redun \le n$, running it on \cref{alg_hard_dist} takes at least $\Omega(\log (n / \Redun))$ expected time per operation.
\end{theorem}
\begin{proof}
  The proof is similar to \cref{sec:hard_instance}, based on a tree structure, but we will use different parameters.

  Same as before, there are $n$ (level-0) leaf nodes on the tree, each corresponds to a single meta-operation. Then, we let each level-1 node represent an interval of $m_1 \defeq \max(R, n^{1 - \gamma / 2})$ consecutive meta-operations, i.e., it is a parent of $m_1$ leaf nodes.\footnote{For simplicity, we assume that all parameters are integers, as is the case throughout this paragraph.} Beyond level 1, we fix $\branch = 2^{64}$ as a large constant and build a $\branch$-ary tree, letting every node (except the root) have exactly $\branch$ children. Finally, there is a single root node at level $h = \log_{\branch}(n / m_1)= \Omega(\log(n/R))$.\footnote{If $m_1 = R$, it is clear that $h = \log_{\branch} (n / R) = \Theta(\log (n/R))$; if $m_1 = n^{1 - \gamma/2}$, we also have $h = \log_{\branch} n^{\gamma / 2} = \Theta(\log n) = \Omega(\log (n/R))$.}

  Recall that $\widl$ represents the number of operations within each level-$\lv$ node. For $1 \le \lv < h$, we have $\widl = m_1 \branch^{\l - 1} \ge \max(R, n^{1 - \gamma / 2})$.

  We apply \cref{lm:outer} on every node in level $\lv \in \left[ 2, \, h \right]$. First, we verify the remaining premises in \cref{lm:outer}:
  \begin{itemize}
  \item $2^{64} \le \branch \le \alpha \log n / 12$. This condition is satisfied as long as $n$ is sufficiently large.
  \item $\widl\log\branch \ge 100\Redun$. For every $\lv \ge 2$, the width $\widl$ satisfies $\widl \log \branch > \branch m_1 > 100R$.
  \item $\frac{1}{2} \gamma \widl \log U \ge \Redun + 5\widl + 2(\beta + 1) \widl t$. We still set $t=\frac{\gamma}{8(\beta + 1)} \log U$, so we only need to show $\frac{1}{4} \gamma \widl \log U \ge \Redun + 5\widl$. When $n$ is sufficiently large, there is
  \begin{align*}
  	\gamma\widl\log U
  	\ge 24\widl
  	\ge 4(\Redun + 5\widl).
  \end{align*}
  \item $\widl\ge n^{1-\gamma/2}$ holds due to the value of $m_1$.
  \end{itemize}
  Therefore, for the levels $\lv \in \left[2, \, h\right]$, we can apply \cref{lm:outer} on every node. Finally, similar to the final step of \cref{thm:main}, we can consider the following two cases to finish the proof:
  \begin{itemize}
  \item If there exists a level $\lv$ in which at least half of the nodes satisfy $\E[\probeu]\ge \frac{1}{16}\widl t$, then we bound the expected time by $\frac{1}{2} \cdot \frac{n}{\widl} \cdot \frac{1}{16}\widl t = \frac{\gamma}{256(\beta + 1)} \cdot n\log U$.  As the coefficient is a constant and $\log U \ge \log (n/\Redun)$, the expected time per operation is at least $\Omega(\log (n/\Redun))$.
  \item Otherwise, for every level $\lv$, at least half of the nodes in level $\lv$ satisfy $\E[\costu]\ge \frac{\gamma}{100}\widl$. Taking summation of all $\E[\costu]$, the total time cost is at least $\sum_{\lv = 2}^{h} \frac{\gamma}{100} \widl \cdot \frac{1}{2} \cdot \frac{n}{\widl} = \frac{\gamma}{200} n(h - 1)$ in expectation, so the expected time cost per operation is at least $\Omega(h) = \Omega(\log(n/\Redun))$. \qedhere
  \end{itemize}
\end{proof}

\section{Key-Value Reduction}
\label{sec:reduction}

\newcommand{\Ckv}{D_{\textup{kv}}}
\newcommand{\Ckonly}{D_{\textup{k-only}}}

So far we have only considered the lower bound for key-only dictionaries that store $n$ different keys in the universe $[U]$. In this section, we will extend this result to dictionaries with values associated with keys:

\begin{definition}
  A \emph{key-value dictionary}, denoted by $\Ckv$, is a dictionary that stores $n$ key-value pairs $(k,v)\in [U]\times[V]$, where all keys are distinct. $\Ckv$ supports insertions and deletions of key-value pairs. Moreover, when querying some key $k$, it returns whether $k$ is present, together with the corresponding value $v$ of $k$ if $k$ is in $\Ckv$.
\end{definition}

In this definition, we call $[U]$ the key universe and $[V]$ the value universe. Without loss of generality, we assume that the word-size $w$ satisfies $w=\log U+\log V$, so we can store both the key and the value in a single word.

The memory usage of $\Ckv$ can be represented by $\log\binom{U}{n}+n\log V+R$, where the first two terms are for succinct storage of $n$ keys and $n$ values, and $R$ is the redundancy of $\Ckv$. Below we will derive a reduction from $\Ckv$ to some key-only dictionary denoted by $\Ckonly$, and prove that we can use $\Ckv$ to simulate $\Ckonly$ with little additional time and space consumption. This theorem implies a time-space lower bound similar to the key-only version.

\begin{theorem}
  \label{thm:key_value_reduction}
  Assume $U=n^{1+\alpha}$ where $\alpha$ is a constant. For sufficiently large $n$, if we have a key-value dictionary $\Ckv$ with redundancy $R=\Omega(n^{1-\alpha / 2})$ and running time $T$, then setting $U'=U\times V$, we can construct a key-only dictionary $\Ckonly$ on universe $[U']$, which can solve Distribution~\ref{alg_hard_dist} with redundancy $R + o(R)$ and running time $T + O(1)$ in expectation.
\end{theorem}

\begin{proof}
  The intuition is to split a key $x\in [U']$ in $\Ckonly$ into two parts $(k,v)$ with length $\log U$ and $\log V$ bits, respectively. We regard the two parts as a key-value pair, and transform the operation on $\Ckonly$ to an operation on $\Ckv$. Roughly speaking, an insertion/deletion of $x$ in $\Ckonly$ is turned into an insertion/deletion of $(k,v)$ in $\Ckv$, and a query of $x$ in $\Ckonly$ will return ``true'' if and only if $k$ is in $\Ckv$ while the corresponding value equals $v$.

  If no key collisions occur, i.e., all pairs $(k, v)$ from the transformation have different $k$, we can directly simulate $\Ckonly$ with $\Ckv$ to complete all operations. In fact, among all $2n$ keys inserted in Distribution~\ref{alg_hard_dist}, the expected number of collisions is only $O(n^2 / U) = O(n^{1-\alpha})$; we maintain a separate compact hash table to store those collided keys, which takes additional $O(n^{1-\alpha} w) = o(n^{1-\alpha/2}) = o(R)$ space and $O(1)$ update/query time.
\end{proof}

When $R = \Theta(n^{1 - \alpha/2})$, the above theorem combined with \cref{thm:ext} already gives the best possible time lower bound $T \ge \Omega(\log (n/R)) = \Omega(\log n)$, which also applies for $R = o(n^{1 - \alpha / 2})$. Therefore, we deduce that the time-space trade-off of key-value dictionaries is not weaker than that of key-only dictionaries. Thus both \cref{thm:main} and \cref{thm:ext} can be applied to show
\begin{itemize}
\item $T \ge \Omega(k)$ when $R = O(n \log^{(k)} n)$;
\item $T \ge \Omega(\log (n/R))$ for all $R \le O(n)$.
\end{itemize}

\section{Update-Only Lower Bound}
\label{sec:no_query}

\newcommand{\va}{\vec v_{\!\smallsub \inseta}}
\newcommand{\vb}{\vec v_{\!\smallsub \insetb}}
\newcommand{\vab}{\vec v_{\!\smallsub \inseta \cup \insetb}}
\newcommand{\vd}{\vec v_{\!\smallsub \delset}}
\newcommand{\vk}{\vec v_{\!\smallsub \keyset}}
\newcommand{\vkd}{\vec v_{\!\smallsub \keyset \setminus \delset}}
\newcommand{\vset}[1]{\vec v_{\!\smallsub #1}}
\newcommand{\vktall}{\vec v_{\!\smallsub \keyset}^{\vphantom{*}}}
\newcommand{\vabtall}{\vec v_{\!\smallsub \inseta \cup \insetb}^{\vphantom{*}}}

In this section, we extend the lower bound from the previous sections by relaxing the restriction on query time and focusing only on the trade-off between space and update time. Specifically, we prove that if the keys are associated with long values, the updates must follow the same time-space lower bound even if the queries are allowed to take arbitrarily long time.

First, note that although \cref{alg_hard_dist} combines a query, a deletion, and an insertion into every meta-operation, the existence of the query is only used once in the inner game. In Step~\ref{step:inner_deleted_set} of the inner game protocol, Bob needs to learn the deleted set $\delset$ by performing queries on $\stabef(\setb^*)$. The inclusion of queries in the meta-operations ensures that the cell set $\setb^*$ is sufficient to answer queries for the true set $\delset$, allowing $\delset$ to pass the test.

This observation suggests that, if we can modify the protocol of the inner game to avoid using queries in Step~\ref{step:inner_deleted_set}, we can eliminate queries from the hard distribution and prove a trade-off between space and update time (without query-time requirements).
We demonstrate below that for the \emph{key-value dictionary} with a relatively long value length, we can make such a modification and prove an update-only lower bound.\footnote{However, the problem of proving an update-only lower bound for a key-only dictionary remains open and will be discussed in \cref{sec:discuss}.}

\UpdateOnly*

The proof of \cref{thm:update_only} is based on directly modifying the proof of \cref{thm_main}, instead of reductions as in \cref{sec:reduction}. It is worth noting that the combination of \cref{thm:key_value_reduction} and \cref{thm:update_only} covers all cases for key-value dictionaries: \cref{thm:key_value_reduction} works for the cases with long keys, i.e., $U=n^{1+\Theta(1)}$, while \cref{thm:update_only} covers the cases with short keys, i.e., $U=n^{1+o(1)}$ and $V=n^{\Theta(1)}$.

\begin{proof}
  The proof of this theorem follows the same framework as the proof of \cref{thm_main}. Similar to Distribution~\ref{alg_hard_dist}, the hard distribution used here is a sequence of $n$ meta-operations, each consisting of a deletion and an insertion only (no query). When we insert a key, its associated value is sampled uniformly at random from $[V]$. We build a tree on top of these $n$ meta-operations and assign each cell-probe to an internal node of the tree as before. We bound the total cost assigned to an internal node using the \emph{outer lemma}.

  The proof of the outer lemma remains the same as in \cref{lm:outer}, which makes use of a communication game (the \emph{outer game}). Before the game starts, we give both Alice and Bob the associated values to all keys. During the game, Alice sends a message to tell Bob each key is inserted in which segment, where the length of the message is analyzed to complete the proof-by-contradiction. The proof of the outer lemma relies solely on the randomness of the \emph{order} of inserted keys, so it works for short keys ($U = n^{1 + o(1)}$) as well.

  It remains to modify the inner lemma to avoid using query operations. The modified inner lemma is shown below. (Some of the constants are different from those in \cref{lm:inner}, but they are not essential to the application in the outer lemma.) Throughout the remainder of this section, we adopt the notations from \cref{sec:inner} unless otherwise specified.

  \begin{lemma}
    \label{lm:inner_no_query}
    For integers $U, V, n, \wid, g$ and real number $\gamma > 0$ satisfying $g < \wid / 2$, $\wid < n / 2$, $U \ge 3n$, $V\ge U^{2 + 5\gamma}/n^2$, assume the random variables $\stabef, K, \dels, \insorder, A, B$ are randomly sampled according to the procedure in \cref{sec:inner}, and assume the associated values of all keys are sampled independently and uniformly at random from $[V]$. Let $t > 0$ be a parameter such that the constraints $\frac{1}{2}\gamma\wid\log U \ge 2\wid t + 2\Redun + 4\wid$ and $\E[|\seta|]\le \wid t$ are satisfied, we must have
    \[ \Pr\!\BigBk{ \inseta \textup{ and } \insetb \textup{ are consistent} \;\Big|\; |\inseta\cap\insetb|=g} \le U^{-\gamma\wid}. \]
  \end{lemma}

  The proof of \cref{lm:inner_no_query} leverages the randomness of the values in addition to the keys. Let $\va$ denote the values associated with keys in the inserted key set $\inseta$. Similarly, we can define $\vb$, $\vd$, $\vk$, $\vkd$, and $\vab$ for $\insetb$, $\delset$, $\keyset$, $\keyset \setminus \delset$, and $\inseta \cup \insetb$, respectively. The inner game is modified as follows by making the key sequences pre-given.
  \begin{itemize}
  \item In addition to the pre-given fixed parameters $g$, $m$, $n$, $U$, and $V$, the sequences of keys $\insas$, $\insbs$, and $\dels$ are also given to both Alice and Bob before the game. (Bob does not know $\keyset$ in advance.)
  \item Alice further knows $\stabef$, $\va$, $\vb$, and $\vd$ before the game.
  \item The goal of the game is to let Bob learn $\vab$, $\vk$, and $\keyset \setminus \delset$ when $\inseta$ and $\insetb$ are consistent.
  \end{itemize}
  The protocol is as follows.
  \begin{enumerate}
  \item Check consistency. Alice sends $\indin \defeq \indicator{A \textup{ and } B \textup{ are consistent}}$ to Bob, and the game terminates if $W = 0$.
  \item Send $\vb$. Alice directly sends $\vb$ using $\wid \log V$ bits.
  \item Send cell contents. Alice sends the artificial memory state $\stamix$ using $Nw = n \log V + \log \binom{U}{n} + \Redun$ bits. Note that Alice does not need to send the partition information by Bloomier Filter as before, since $\setb$ can be learned by Bob in the next step.
  \item Recover $\setb$ and $\stab(\setb)$. As Bob already knows $\insbs$, $\vb$ and $\dels$, he can simulate the meta-operation sequence of $\insbs$ on $\stamix$. (Note that deleting a key from the data structure only requires knowledge of the key but not its value, thus Bob does not need $\vd$ to complete the simulation.) We still let $\setb$ represent the set of cells probed during this simulation, which Bob can learn via the simulation. Since $\stamix(\setb) = \stabef(\setb)$, the simulation process is the same as doing these operations on $\stabef$, so Bob can also learn $\stab(\setb)$, the memory state after the simulation.
  \item Recover $\va, \vkd$ and $\keyset \setminus \delset$. Alice sends $\seta \cap \setb$ (conditioned on $\setb$) together with the difference between $\staa(\seta \cap \setb)$ and $\stab(\seta \cap \setb)$ using $2\wid t + \gamma \wid \log U$ bits in expectation. Then, Bob can learn $\staa$ by combining $\staa(\seta \cap \setb)$ and $\stamix(\setnotacapb)$. He can further extract $\va, \vkd$, and $\keyset \setminus \delset$ from $\staa$.\footnote{Extracting information from the data structure can be done by querying all possible elements in the key universe. This can take arbitrarily long time.}

  \item Recover $\vd$. This step is different from \cref{sec:inner} in the sense that we do not rely on query operations here.
    \begin{itemize}
    \item First, Alice sends $\stab$: Since all the keys stored in $\stab$ (i.e., $\keyset \setminus \delset$ and $\insbs$) and their corresponding values ($\vkd$ and $\vb$) are known to Bob, Alice only needs to send $\stab$ \emph{conditioned on} these keys and values. This can be done using at most $\Redun$ bits.
    \item Next, Bob computes $\stabef$: It is obtained by combining $\stabef(\setb) = \stamix(\setb)$ and $\stabef(\setnotb) = \stab(\setnotb)$. The first equality holds due to the definition of $\stamix$, while the second holds because cells in $\setnotb$ are not probed in the process of transforming $\stabef$ to $\stab$.
    \item Finally, Bob extracts $\vd$ from $\stabef$.
    \end{itemize}
  \end{enumerate}

  Similar to \cref{lm:inner}, we can derive the following inequality from the modified inner game: 
  \begin{align*}
    \label{eq:no_query_inner}
    H(\msgfrom \mid \indin = 1) \;\ge\; H(\msgto \mid \indin = 1). \numberthis
  \end{align*}
  To compute the left-hand side, we add up the messages sent in each step:
  \begin{align*}
    H(\msgfrom \mid \indin = 1) \;\le\; \wid \log V + \bk{n \log V + \log \binom{U}{n} + \Redun} + \bk{2\wid t + \gamma \wid \log U} + R.
  \end{align*}
  For the right-hand side, we can apply the same technique as in \cref{lm:inner} to handle the condition $\indin = 1$. Here, $\msgto = (\vab, \vk, \keyset \setminus \delset \mid \insas, \insbs, \dels)$, so we can obtain
  \begin{align*}
    & H(\msgto \mid \indin = 1)
    \\
    & \ge \log \bk{1 \middle/\! \max_{\vab^*, \vk^*, \keyset^*} \; \Pr_{\vabtall, \vktall, \keyset} \!\BigBk{(\vab, \vk, \keyset \setminus \delset) = (\vab^*, \vk^*, \keyset^* \setminus \delset) \;\Big|\; \insas, \insbs, \dels, \indin = 1}} \\
    & \ge \log \bk{\Pr[\indin = 1] \middle/\! \Pr_{\vabtall, \vktall, \keyset} \!\BigBk{(\vab, \vk, \keyset \setminus \delset) = (\vab^*, \vk^*, \keyset^* \setminus \delset) \;\Big|\; \insas, \insbs, \dels}} \\
    & = \log \Pr\Bk{W = 1} + \bk{2m - g + n} \log V + \log \binom{U - 3m + g}{n-m} \\
    & \ge \log \Pr\Bk{W = 1} + \bk{2m - g + n} \log V + \log \binom{U - 3m}{n-m}.
  \end{align*}
  Plugging these bounds back into \eqref{eq:no_query_inner}, we get
  \begin{align*}
    \log \frac{1}{\Pr \Bk{W = 1}} 
    \ge{}& (\wid - g) \log V + \log \binom{U - 3m}{n - \wid} - \log \binom{U}{n} - \bk{\gamma \wid \log U + 2\wid t + 2R} \\
    \ge{}& \frac{1}{2}\wid \log V - \bk{\wid \log \frac{U}{n} + 4 \wid} - \bk{\gamma \wid \log U + 2\wid t + 2R}\\
    \ge{}& \frac{5}{2}\gamma\wid \log U - \bk{\gamma \wid \log U + 2\wid t + 2R + 4m}\\
    \ge{}& \gamma\wid \log U.
  \end{align*}
  Here the second inequality is due to $g < m / 2$ and
  \begin{itemize}
  \item $\log \binom{U}{n - m} - \log \binom{U - 3m}{n - m} \le (n - m) \log \frac{U - (n - m)}{U - 3m - (n - m)} \le (n - m) \log \frac{2n + m}{2n - 2m} \le \frac{3m}{2 \ln 2} < 3m$ (here we used $\log (1 + x) \le x / \ln 2$);
  \item $\log \binom{U}{n} - \log \binom{U}{n - \wid} \le \wid \log \frac{U - n + \wid}{n - \wid + 1} \le \wid \log \frac{U}{n/2} = \wid \log \frac{U}{n} + \wid$.
  \end{itemize}
  The third inequality is due to the condition $V \ge U^{2 + 5\gamma}/n^2$, and the fourth inequality is due to the condition $\frac{1}{2}\gamma \wid \log U \ge 2\wid t + 2R + 4m$. This proves \cref{lm:inner_no_query}.
\end{proof}

\section{Lower Bounds for Related Problems}
\subsection{Strongly History-Independent Dictionaries}
\label{sec:hist_ind}
In this subsection, we show a brief overview of the lower bounds for strongly history-independent dictionaries. Recall that a strongly history-independent dictionary's memory state only depends on the current set of keys stored in it, and possibly some random bits; moreover, by Yao's Minimax Principle, we may assume without loss of generality that the algorithm is deterministic, which means we can fully recover its memory state by knowing only the current key set.

This fact benefits our Protocol~\ref{ptc:outer} for the outer game: Bob already knows the starting state $\stabegin$, the keys to insert $\insaset$, and the keys to delete $\dels$, which together can infer the key set $\keyend$ at the end, and further, the memory state $\staend$ at the end. Hence, Alice no longer needs to send $\staend$ to Bob, which costed $R$ bits of information and was the only step involving the redundancy $R$ in \cref{sec:outer}. After eliminating this cost, the proof in \cref{sec:outer} works regardless of the redundancy $R$, except that it still relies on the Inner Lemma~\ref{lm:inner}.

We do not change the statement (or the proof) of the inner lemma; however, it has a larger tolerance of redundancy $R$ than the initial outer lemma: $R \le \frac{1}{100} \gamma \widl \log U = \Theta(\widl \log U)$ suffices for its premise, where $\widl$ is the number of meta-operations the current node represents.

Recall that the entire proof is based on a tree on top of $n$ meta-operations. Now, we set the tree parameters similarly to \cref{sec:extend}: the branching factor is a fixed large constant $\branch = 2^{64}$, while every level-1 node (parent of leaves) represents $m_1 \defeq \max\bigbk{\frac{100 R}{\gamma \log U}, \, n^{1 - \gamma / 2}}$ consecutive meta-operations. Under these parameters, the height of the tree is $\Theta\bigbk{\log \frac{n \log U}{R}}$, and for every internal node of the tree,
\[
  \frac{1}{100} \gamma \widl \log U \ge \frac{1}{100} \gamma m_1 \log U \ge R,
\]
which means the inner lemma's premise is satisfied, thus the outer lemma holds on levels $\l \ge 2$. By a similar argument as the proof of \cref{thm:main}, we conclude the following result.

\HistIndLB*

\subsection{Stateless Allocation}
\label{sec:stateless}

We recall the \emph{stateless allocation} problem mentioned in \cref{sec:intro}: The algorithm is given a set $S \subseteq [U]$ of at most $(1-\eps) n$ elements, and it should allocate these elements to $n$ slots $\{1, 2, \ldots, n\}$, where each slot can accommodate at most one element. The allocation should be an injection from $S$ to $[n]$ that only depends on the current set $S$ as well as random bits $r$ that are fixed over time, written $\sigma_{S, r} : S \to [n]$. When an insertion/deletion changes $S$ to $S'$, we define the expected switching cost to be
\[
  \E_r\bigBk{\bigabs{\midBK{x \in S \cap S' \mid \sigma_{S, r}(x) \ne \sigma_{S', r}(x)}}},
\]
which equals the number of elements that change their assigned slots during the update. This problem is very similar to the \emph{slot model} for dictionaries introduced in \cref{sec:overview}, where we have to assign $n$ keys to $n$ slots, with the following minor differences:
\begin{itemize}
\item Stateless allocation allows $\eps$-fraction of the slots to be empty, while the slot model utilizes all $n$ slots.
\item Slot model additionally allows the mapping from keys to slots to be determined by not only the current key set but also $O(n)$ bits of redundancy, making it not strongly history-independent.
\end{itemize}
Below, we first show an $\Omega(\log n)$ lower bound on the expected switching cost of any stateless allocation algorithm with $\eps = 0$, which one can think of as the slot model with 0 redundancy.

\begin{proof}[Proof Sketch for $\eps = 0$]
  The proof structure is again similar to that of \cref{thm:main}: We build a tree structure over $n$ meta-operations each consisting of one deletion followed by one insertion, where the branching factor $\lambda = 2^{64}$ is a fixed constant for every node, which implies that the height of the tree equals $\Theta(\log n)$. When some key (element) is moved in two meta-operations $t_1 < t_2$ but not in between, we add one cost to the LCA of these two meta-operations on the tree. Then, using almost the same argument as \cref{sec:overview}, we are able to prove that the expected cost on every node $u$ that represents $m_\l$ meta-operations is at least $\Omega(m_\l)$. Note that the branching factors used here are smaller than those we used in \cref{sec:overview} to prove the $\Omega(\log^* n)$ lower bound for the slot model, and the proof still works because the redundancy $R = 0$. Finally, we sum up the cost on all internal nodes and conclude an $\Omega(\log n)$ lower bound.
\end{proof}

\newcommand{\virt}{\bot}

When $\eps > 0$, there are $\eps n$ slots left empty. We add $\eps n$ placeholder elements $\virt_1, \ldots, \virt_{\eps n}$ and put $\virt_i$ into the $i$-th empty slot, which leads to an algorithm allocating $S \cup \midBK{\virt_1, \ldots, \virt_{\eps n}}$ to all $n$ slots, without leaving any slot empty, i.e., the new problem with placeholders is a stateless allocation problem with $\eps = 0$.

We adapt the proof for $\eps = 0$ to the new instance with placeholders. Again, we build a $\lambda$-ary tree over $n$ meta-operations, and when some element is moved in two meta-operations $t_1 < t_2$ but not in between, we add one cost to the LCA of these two leaves: This cost is said to be a \emph{real cost} if the moved element is a real element in $S$, or a \emph{virtual cost} if the moved element is a placeholder. The proof of $\eps = 0$ tells us that, for an internal node $u$ representing $m_\l$ meta-operations, the real and virtual cost add up to $\Omega(m_\l)$. However, only real cost will cause switching cost of the initial stateless allocation problem.

Fortunately, the virtual cost on any node $u$, which is maximized when every placeholder is moved in all $\lambda$ segments (children), cannot exceed $\lambda \eps n$. When $m_\l \ge c \eps n$ for some large constant $c$, the sum of virtual and real cost is at least $\Omega(m_\l) \ge 2 \lambda \eps n$, which implies that the real cost on this node is at least $\Omega(m_\l) - \lambda \eps n \ge \frac{1}{2} \cdot \Omega(m_\l) = \Omega(m_\l)$. This inequality can apply to the top $\Theta(\log \eps^{-1})$ levels of the tree where $m_\l \ge c \eps n$, thus we conclude an $\Omega(\log \eps^{-1})$ lower bound on the switching cost, as stated in the following theorem.

\AllocLB*

(Similar to \cref{sec:no_query}, $U \ge 3n$ suffices for the proof, because the outer communication game only relies on the randomness of the \emph{order} of inserted keys rather than the keys themselves.)

\section{Discussions}
\label{sec:discuss}

We have proved tight bounds for dynamic succinct dictionaries in the previous sections.
Now let us discuss the limitations of our methods.

\paragraph{The case when \texorpdfstring{$U=n^{1+o(1)}$}{U=n\^{1+o(1)}}.} Recall that all our conclusions are based on the assumption $U=n^{1+\alpha}$ for $\alpha=\Theta(1)$. If we allow slightly subconstant $\alpha$ here, the proof of \Cref{thm:main} can only prove a bound of $\Omega(\alpha k)$ time when the wasted bits per key is $\log^{(k)} n$. This multiplicative factor $\alpha$ comes from \Cref{lm:outer} which gives the proposition $\costu \ge \gamma m/100$. When $\alpha$ was a constant, we could regard $\gamma$ as a constant and sum up all $\costu$ to prove a time lower bound $\Omega(k)$; however, when $\alpha=o(1)$, $\gamma = \Theta(\alpha)$ is no longer a constant, so we can only obtain the bound $\Omega(\alpha k)$. This still implies an $\omega(1)$ time lower bound when $\alpha = \omega(1 / \log^* n)$ under $O(1)$ wasted bits per key.
In contrast,~\cite{liu22} showed that when $\alpha=1/\log^{(t)} n$ for any constant $t$, one can achieve $o(1)$ wasted bits per key with constant running time. So the $\Omega(\alpha k)$ lower bound cannot be improved much. Similarly, when $R < n$, \Cref{thm:ext} is weakened to $\Omega\bk{\min\BK{\alpha \log (n/R), \alpha^2 \log n}}$.\footnote{The second $\alpha$ factor comes from the tree depth $h \approx \min\{\log (n/R), \log n^{\gamma/2}\} = \Theta(\min\BK{\log (n/R), \alpha \log n})$.}

\paragraph{Key-only dictionaries with no query time constraints.} In previous sections, we have obtained the time-space trade-off for key-only dictionaries with query time requirements, as well as key-value dictionaries without query time requirements. However, the problem of proving a lower bound for key-only dictionaries without query time constraints still remains open. Below is a simple example that shows why our method does not work for this goal.

Recall that in \Cref{lm:inner_re}, we claim that for two $m$-sized sets $A,B$, the probability that $C_A(S_A\cap S_B)$ and $C_B(S_A\cap S_B)$ has at most $\gamma m$ different cells is at most $U^{-\gamma m}$. The proof of this lemma makes use of the queries. If there are no queries in the meta-operations, we can construct a data structure that contradicts \Cref{lm:inner_re}:

Suppose we use three cells $C_1,C_2,C_3$ to maintain keys $x_1, x_2, x_3 \in [0, 5U)$, allowing $O(1)$-bit redundancy. If $x_1\in[0,3U),x_2\in[3U,4U),x_3\in[4U,5U)$ (which happens with constant probability, and we call it the \emph{good case}), we organize the keys according to the three cases listed in \cref{table:keyonly_counter_ex}. Otherwise, we give up, succinctly store all keys, and check all cells for every operation. We store in an $O(1)$-bit extra memory whether the good case applies, and if so, which case in \cref{table:keyonly_counter_ex} applies. It is easy to see that the whole data structure incurs $O(1)$ bits of redundancy.

\begin{table}[H]
  \centering
  \caption{Data structure organization under the good case.}
  \label{table:keyonly_counter_ex}
  \begin{tabular}{|l|c|c|c|}
    \hline
    & $C_1$ & $C_2$ & $C_3$ \\
    \hline
    Case 1: $x_1 \in [0, U)$ & $x_2 \oplus x_1$ & $x_3$ & $x_2 \oplus x_3$ \\
    Case 2: $x_1 \in [U, 2U)$ & $x_2$ & $x_3 \oplus x_1$ & $x_2 \oplus x_3$ \\
    Case 3: $x_1 \in [2U, 3U)$ & $x_2$ & $x_3$ & $x_2 \oplus x_3 \oplus x_1$ \\
    \hline
  \end{tabular}
\end{table}

\newcommand{\keywithname}[1]{x_{\textup{#1}}}

We let $m = 1$ in \cref{lm:inner_re}, i.e., there is only one meta-operation. We randomly delete a key and consider two possible keys to be inserted, denoted by $\keywithname{A}$ and $\keywithname{B}$ (they correspond to $A, B$ in \cref{lm:inner_re}). With constant probability, we will be deleting $x_1$ and inserting back two new $x_1$'s ($\keywithname{A}, \keywithname{B} \in [0, 3U)$), which means the good case still applies after both possible insertions. In this case, we denote the initial $x_1$ by $\keywithname{1D}$, and denote the two inserted keys by $\keywithname{1A} \defeq \keywithname{A}$ and $\keywithname{1B} \defeq \keywithname{B}$.

While $\keywithname{1D}, \keywithname{1A}, \keywithname{1B}$ are selected uniformly at random, there is a constant probability that $\keywithname{1D}$, $\keywithname{1A}$, $\keywithname{1B}$ lead to Case 2, 1, and 3, respectively. Suppose we are deleting $\keywithname{1D}$ and inserting $\keywithname{1A}$, we can complete the meta-operation by probing only $C_1$ and $C_2$:
\begin{itemize}
\item $C_2$ initially stores $x_3 \oplus \keywithname{1D}$. We read $\keywithname{1D}$ from the deletion operation itself and read $x_3$ by probing $C_2$.
\item We probe $C_1$ to read $x_2$.
\item After knowing all three keys $\keywithname{1A}, x_2, x_3$ that should be stored after the insertion, we directly write $x_2 \oplus \keywithname{1A}$ and $x_3$ into $C_1$ and $C_2$, respectively. (Recall that Case 1 is the target case.)
\end{itemize}
Similarly, when we are deleting $\keywithname{1D}$ and inserting $\keywithname{1B}$, we can probe only $C_2$ and $C_3$. The only commonly probed cell is $\seta \cap \setb = \{C_2\}$, whose content will be the same after both processes, which means that $\staa(\seta \cap \setb)$ and $\stab(\seta \cap \setb)$ have a constant probability to be equal. With a sufficiently large $U$, this example contradicts \cref{lm:inner_re}.\footnote{Strictly speaking, it contradicts the following variant of \cref{lm:inner_re}: An $O(1)$-bit extra memory can be accessed which is not taken into account in the definition of consistency, and we only require $U$ (but not $n$) to be sufficiently large. If the meta-operations include queries, then this variant can also be proved similarly to the proof of \cref{lm:inner_re}.}
Thus, to prove a key-only lower bound without query time constraints, we need to develop a new method that does not rely on \cref{lm:inner_re}.

\bibliographystyle{alpha}
\bibliography{reference.bib}

\end{document}